\documentclass[9pt,journal]{IEEEtran}

\usepackage{amsmath}
\usepackage{amssymb}
\usepackage{graphicx}
\usepackage{amsthm}
\usepackage{mathtools}
\usepackage{bm}
\usepackage{xcolor}
\usepackage{hyperref}
\usepackage{float}

\usepackage{threeparttable}
\usepackage{array}
\newcolumntype{P}[1]{>{\centering\arraybackslash}p{#1}}

\usepackage[inline]{enumitem}

\usepackage{cite}

\newtheorem{theorem}{Theorem}
\newtheorem{definition}[theorem]{Definition}

\theoremstyle{remark}

\ifCLASSINFOpdf

\else

\fi

\begin{document}

\title{Twisting Signals for Joint Radar-Communications: \\An OAM Vortex Beam Approach}

\author{Wanghan Lv, \IEEEmembership{Member, IEEE}, Kumar Vijay Mishra, \IEEEmembership{Senior Member, IEEE} and Jinsong Hu, \IEEEmembership{Member, IEEE}
\thanks{This work was supported in part by the Nature Science Foundation of China (NSFC) under Grant 62441109 and Grant 62001116 and in part by Natural Science Foundation of Jiangsu Higher Education Institutions of China under Grant 23KJB510008. The conference precursor of this work was presented at the 2023 IEEE International Workshop on Signal Processing Advances in Wireless Communications (SPAWC) \cite{lv2023joint}. }
\thanks{W. L. is with College of Computer and Information Engineering, Nanjing Tech University, Nanjing 211816, China, e-mail: lwanghan@njtech.edu.cn. }
\thanks{K. V. M. is with The University of Maryland, College Park, MD 20742 USA, e-mail: mishra@umd.edu.}
\thanks{J. H. is with Fuzhou University, Fuzhou, China, e-mail: jinsong.hu@fzu.edu.cn.}
}

\maketitle

\IEEEpeerreviewmaketitle

\begin{abstract}
Orbital angular momentum (OAM) technology has attracted much research interest in recent years because of its characteristic helical phase front twisting around the propagation axis and natural orthogonality among different OAM states to encode more degrees of freedom than classical planar beams. Leveraging upon these features, OAM technique has been applied to wireless communication systems to enhance spectral efficiency and radar systems to distinguish spatial targets without beam scanning. Leveraging upon these unique properties, we propose an OAM-based millimeter-wave joint radar-communications (JRC) system comprising a bi-static automotive radar and vehicle-to-vehicle (V2V) communications. Different from existing uniform circular array (UCA) based OAM systems where each element is an isotropic antenna, an OAM spatial modulation scheme utilizing a uniform linear array (ULA) is adopted with each element being a traveling-wave antenna, producing multiple Laguerre-Gaussian (LG) vortex beams simultaneously. Specifically, we first build a novel bi-static automotive OAM-JRC model that embeds communication messages in a radar signal, following which a target position and velocity parameters estimation algorithm is designed with only radar frames. Then, an OAM-based mode-division multiplexing (MDM) strategy between radar and JRC frames is presented to ensure the JRC parameters identifiability and recovery. Furthermore, we analyze the performance of the JRC system through deriving recovery guarantees and Cram\'er-Rao lower bound (CRLB) of radar target parameters and evaluating the bit error rate (BER) of communication, respectively. Our numerical experiments validate the effectiveness of the proposed OAM-based JRC system and parameter estimation method.
\end{abstract}

\begin{IEEEkeywords}
Automotive radar, joint radar-communication, Laguerre-Gaussian vortex beam, orbital angular momentum, parameter estimation.
\end{IEEEkeywords}

\section{Introduction}
In recent years, autonomous vehicles have been envisioned to revolutionize transportation and are thus the focus of increasing interest both in academia and industry \cite{VKKukkala2018,SSaponara2019}. Such self-driving cars are expected to navigate efficiently and safely in a wide variety of complex uncontrolled environments, promoting the development in vehicular control \cite{CMarina2018}, environmental sensing \cite{LXu2023}, and efficient resource utilization \cite{LWu2022}. To increase automotive safety, self-driving cars are equipped with multiple sensors, including lidar, camera, infrared detectors, global navigation satellite system, and radar transceivers. Among all these sensors, radar provides the ability to accurately detect distant objects and offers the advantage of robust detection in adverse vision and weather conditions compared to other competing sensing technologies \cite{SMPatole2017}. In addition to environmental sensing, autonomous vehicles are also required to carry out various forms of communications, such as vehicle-to-vehicle (V2V), vehicle-to-infrastructure (V2I), vehicle-to-pedestrian (V2P), and vehicle-to-network (V2N) transmissions, which is the so-called vehicle-to-everything (V2X) strategy \cite{MNoor2022}. Considering that automotive radar and communication systems are limited in size, power, weight, and cost \cite{JLee2022}, the synergistic design of these functionalities is necessary and such joint radar-communications (JRC) systems are increasingly investigated for spectrum and hardware sharing techniques, which have advantages of low cost, compact size, less power consumption, efficient spectrum utilization, and improved performance \cite{KVMishra2019,mishra2024signal}. 

Since automotive JRC systems implement both radar and communications using a single device, the goal of such systems is to transmit a single waveform and employ efficient receiver processing to extract both radar and communications parameters \cite{SHDokhanchi2019}. In general, existing JRC systems which adopt a single waveform can be divided into three main categories, namely radar waveform-based schemes, communications waveform-based approaches, and joint dual-function waveform designs \cite{JAZhang2022}. The radar waveform-based scheme is a radar-centric design that realizes the communication function in a primary radar system. Traditionally, either pulsed or continuous-wave radar signals are adopted, where the communication information can be embedded into the phase of the radar waveform \cite{CSahin2017}, the sidelobe \cite{AHassanien2016_2}, the carrier frequency \cite{XWang2019} or the basis of the radar sub-pulses \cite{RXu2023}. Recently, to improve communication rates, radars with advanced waveforms such as orthogonal frequency-division multiplexing (OFDM) \cite{GHakobyan2020} and frequency hopping \cite{THuang2020} are also emerging, which can embed more information bits into radar signals. For communications waveform-based approaches, namely communication-centric design, the research focus is on how to realize the radar sensing function in a primary communication system. In the millimeter-wave (mmWave) band, the IEEE 802.11 signals have been widely used because of high throughput advantages arising from wide bandwidth. For example, the IEEE 802.11 and IEEE 802.11p OFDM communication waveforms are used for radar sensing in vehicular networks \cite{RCDaniels2018,PKumari2018}. Unlike the previous two categories, the third category considers the design and optimization of the signal waveform without bias to either communication or radar sensing, which aims to fulfill the desired applications only. Existing classical examples include high-frequency systems that can potentially achieve both high data-rate communications and high-accuracy sensing \cite{KVMishra2019,YLuo2020,AMElbir2021} and multi-channel JRC systems which can offer an overall large signal bandwidth for sensing without increasing instantaneous communication bandwidth \cite{FDong2021}. Note that the above JRC techniques are all based on planar electromagnetic (EM) wave transmission, which poses several technical challenges for practical application. For example, in automotive scenarios, common automotive targets, e.g. pedestrians, bicycles, and cars, show distinctive micromotion features, such as the rotation of wheels and the swinging motions of arms and legs \cite{GDuggal2020}. Thus, a specific problem is how to detect the microfeatures of targets. Moreover, for existing wireless communication systems, another thorny problem is how to further improve the capacity and spectrum efficiency. 

As mobile networks evolve toward higher-generation radio frequency (RF) systems, orbital angular momentum (OAM) technology, due to its spiral wavefront, theoretically infinite states and natural orthogonality among different states, has received much attention in recent years \cite{WCheng2019}. Compared with the conventional degrees of freedom (DoFs) in wireless communications, such as DoFs in time, frequency and space domain, the OAM technology provides a new degree of freedom, namely OAM state, which can be utilized for mode-division multiplexing (MDM) of signals \cite{willner2021perspective}. This has the potential to yield a dramatic increase in channel capacity without requiring additional spectrum. 
In \cite{XGe2017}, a new OAM spatial modulation (OAM-SM) mmWave communication system was proposed and energy efficiency can be improved by approximately three times compared to the existing millimeter wave multiple-input-multiple-output (MIMO) communication system. Furthermore, the influences of atmospheric turbulence and misalignment between transmitting and receiving antennas on the OAM-SM transmission system were investigated \cite{XXiong2022}. To enhance the spectral efficiency (SE), OAM multiplexing can also be used simultaneously with a multiplexing technique that uses the orthogonality of spin angular momentum, known as polarization, and frequency- and time-domain multiplexing techniques. For example, in view of the advantage of OFDM technology in spectral efficiency, OAM technology was combined with OFDM technology, leading to an OAM-OFDM wireless communication system \cite{XXiong2020,LLiang2020}. In \cite{HSasaki2024}, OAM multiplexing was extended to OAM-MIMO multiplexing technology, which effectively combines the advantage of OAM multiplexing with that of MIMO-based digital signal processing with multiple uniform circular arrays for line-of-sight wireless transmission. 

For radar sensing, the vortex EM wave carrying OAM has been shown to achieve angular diversity without relative motion or beam scanning \cite{RChen2018,bu2022vortex}. Therefore, in target detection realms, OAM-based scheme can be applied to achieve the azimuthal profile of targets with high resolution \cite{KLiu2018}. Moreover, the spinning object can be detected using the rotational Doppler phenomenon when the line of sight is perpendicular to the target \cite{YWang2021}. Although the vortex EM wave brings the advantage of azimuthal super-resolution, its elevation resolution is restricted. To solve the problem, a three-dimensional (3-D) forward-looking imaging method with EM vortex was proposed, which breaks the limit of elevation resolution in conventional EM vortex \cite{JWang2022}. In \cite{DLiu2023}, an enhanced forward-looking bistatic imaging scheme based on OAM mode design was studied, where higher azimuth resolution performance can be achieved by extending Doppler bandwidth. In order to unify the radar sensing and wireless communications with a dual function OAM waveform, a communication and radar system that exploits OAM has been designed to increase data rates and radar sensitivity to certain chiral targets \cite{DOrfeo2021}. Very recently, \cite{WXLong2023} proposed a novel uniform circular array (UCA) based radar-centric JRC scheme including the OAM-based 3-D position estimation, rotation velocity detection, and specific target communication. 

Previous OAM-JRC studies mainly focused on monostatic sensing applications. However, vehicular JRC is usually bi- or multi-static \cite{SHDokhanchi2019}, wherein the application of OAM requires further investigation. In particular, the bi-static radar exploits the signals reflected from other vehicles and benefits from extending the sensing area to the non-line-of-sight positions relative to the receiver. In this paper, contrary to existing OAM-JRC systems, we mainly focus on the investigation of a bi-static mmWave OAM-JRC system in automotive scenario, where an OAM-based MDM strategy is designed to ensure the parameters identifiability. Preliminary results of this work appeared in our conference publication \cite{lv2023joint}, which introduced the preliminary conceptual ideas of bi-static OAM-JRC. The current journal paper presents a comprehensive unified description of the bi-static OAM-based JRC model and JRC parameters estimation algorithm, as well as provides complete theoretical analysis with detailed proofs, and extensive simulation study. We summarize the main contributions of this paper as: 

\textbf{1) Twisting JRC waveforms for bi-static automotive scenarios.} For the bi-static automotive scenario, classic JRC signals such as phase-modulated-continuous-wave (PMCW) and orthogonal-frequency-division-multiple-access (OFDMA) waveforms are under the framework of planar EM, where time, frequency, and space domains can be exploited for radar and communications encoding. In this paper, the vortex EM wave carrying OAM states is utilized, whose helical phase wavefront results in a twisting waveform. This twisting structure provides a new degree of freedom for JRC application, which can enhance the spectral efficiency of communications and azimuthal resolution of radar sensing. Inspired by the OAM-SM mmWave communication system \cite{XGe2017}, we present a uniform linear array (ULA) based OAM-MIMO JRC scheme for bi-static automotive scenarios, where each element is equipped with a travelling-wave antenna producing multiple OAM waves independently. Compared to classic UCA-based OAM systems, our proposed scheme shows the advantages in system simplification and energy efficiency. 

\textbf{2) An OAM-based MDM scheme to recover JRC parameters.} Different from traditional multiplexing strategies which focus on time, frequency, or space, an OAM-based MDM scheme is explored in this paper, which takes advantage of the orthogonality of vortex waveforms with different OAM modes. Specifically, we adopt the OAM-based MDM between radar and JRC frames. To recover the target position parameters and communication symbols with computational efficiency, a joint estimation of signal parameters via rotational invariance technique (ESPRIT) algorithm is devised. Our proposed algorithm can overcome the well-known paring issue among target position parameters and communication symbols effectively. In addition, the OAM-based target velocity estimation, including linear velocity and rotation velocity, is also discussed in our proposed OAM-JRC bi-static system.

\textbf{3) JRC parameters identifiability and statistical performance bounds.} As known to us all, parameters identifiability is the ability of recovering unknown parameters with limited number of measurements. In this paper, considering that the radar and JRC frames are multiplexed in the OAM mode domain through a sharing factor $\mu$, we show that the JRC parameters identifiability is determined by the sharing factor $\mu$, the number of transmit (receive) array elements, and the number of OAM states. Moreover, the detailed theoretical recovery guarantee is derived to uniquely determine the unknown radar and communication parameters. To investigate the JRC performance, we also analyze the Cram\'er-Rao lower bound (CRLB) of radar target parameters and bit error rate (BER) of communication symbols in the radar and communication subsystems, respectively. 

The remainder of this paper is organized as follows. In the next section, we formulate the signal model of the proposed bi-static OAM-JRC system. In Section \ref{sec:algorithm}, we present an OAM-based MDM between the radar and JRC frames and develop the JRC parameters recovery algorithms. After that, we make a detailed theoretical analysis about the OAM-JRC system by deriving recovery guarantees and CRLB of radar target parameters estimaion and evaluating the BER of communications in Section \ref{sec:PA}. In Section \ref{sec:SE}, we demonstrate
the performance of our proposed method through numerical examples. Finally, Section \ref{sec:summary} concludes this article.

    \begin{table*}
    \renewcommand{\arraystretch}{1.5}
    \caption{Comparison of different OAM beam modes}
    \label{tbl:comp1}       
    \centering
    \begin{threeparttable}
    \begin{tabular}{p{1.5cm}p{3.0cm}P{5cm}p{3.0cm}p{1.2cm}P{2cm}}
    \hline\noalign{\smallskip}
    q.v. & OAM Beam Mode & Field Distribution & Generation & Cost & Application 
    \\
    \noalign{\smallskip}
    \hline
    \noalign{\smallskip}
    \cite{YWang2017} & Perfect vortex beam  & $\delta(r-r_0) \mathrm{exp}(-jl\varphi)$ & Ideal & N/A & N/A \\
    \cite{HWu2013}  & Gaussian vortex beam                &   $\mathrm{exp}\left(-\frac{r^2}{w^2} \right) \mathrm{exp}(-jl\varphi) $     & Spiral phase plate (SPP), computer-generated holograms (CGH), metasurface  &     Low     & Free space optical (FSO) communication  \\
    \cite{FGori1987} & BG vortex beam & $J(\alpha r) \mathrm{exp}\left(-\frac{r^2}{w^2} \right) \mathrm{exp}(-jl\varphi)  $ & UCA &  High   & FSO/wireless communication, radar \\
    \cite{HZhang2022}  & Elliptic vortex beam  & $ \mathrm{exp}\left[-\frac{x^2 + (\epsilon y)^2}{w^2}  \right] \left[ \frac{\sqrt{x^2 + (\epsilon y)^2}}{w} \right]^{\vert l \vert} \cdot  \mathrm{exp} \left[-jl \arctan \left(\frac{\epsilon y}{x} \right) \right]   $   & Spatial light modulator (SLM) & Low &  FSO communication  \\
    \cite{BNdagano2018}   & Vector vortex beam  & $\cos (\theta) \mathrm{exp}(jl\varphi) \hat{R} + \sin (\theta) \mathrm{exp}(-jl\varphi + j\gamma) \hat{L}$  & SLM, metasurface  &  Low  & FSO communication \\
    This paper & LG vortex beam & Eq. (\ref{eq:Prime_1}) & SPP, travelling-wave antenna &     Low     & FSO/wireless communication, radar  \\
    
    \noalign{\smallskip}\hline\noalign{\smallskip}
    \end{tabular}
    \end{threeparttable}
    \end{table*}
    
Throughout this paper, we represent the transpose, Hermitian and conjugate by $(\cdot)^T$, $(\cdot)^H$ and $(.)^*$, respectively. $\otimes$, $\odot$ and $\diamond$ denote the Kronecker, Khatri-Rao and Hadamard products, respectively. $\mathrm{E}[.]$ is the statistical expectation function. $P(.)$ denotes the probability of an event. $\mathrm{Re}\{.\}$ returns the real part. $\mathbf{I}_L$ is the $L\times L$ identity matrix. $\mathbf{1}_{N\times L}$ is the $N\times L$ matrix of all ones. The $(i,j)$-th entry of matrix $\mathbf{A}$ is $[\mathbf{A}]_{i,j}$ and the $j$-th column of $\mathbf{A}$ is $[\mathbf{A}]_j$. $\angle (.)$ denotes the phase of it argument, and $\mathrm{vec}(.)$ is the vectorization operator that turns a matrix into a vector by stacking all the columns on top of another. $\mathrm{diag}(\mathbf{a})$ returns a square diagonal matrix with the elements of vector $\mathbf{a}$ on the main diagonal. $\mathrm{flipud}(.)$ returns a matrix with its rows flipped in the up-down direction.

\section{Signal Model}   \label{sec:SM}
Recall that an EM wave carries two types of momenta, namely linear momentum and angular momentum \cite{RChen2023,CAn2023}. The angular momentum can be further decomposed into spin angular momentum (SAM) and OAM. The SAM, namely polarization, has been studied thoroughly, which describes merely the direction of oscillation in the electric field. Different with SAM, the OAM component stems from the variation of wave phase with respect to the azimuthal angle around the propagation axis of the wave. 

\subsection{Preliminaries of Laguerre-Gaussian Beams}    

\begin{figure}[t]
\centerline{\includegraphics[scale=0.62]{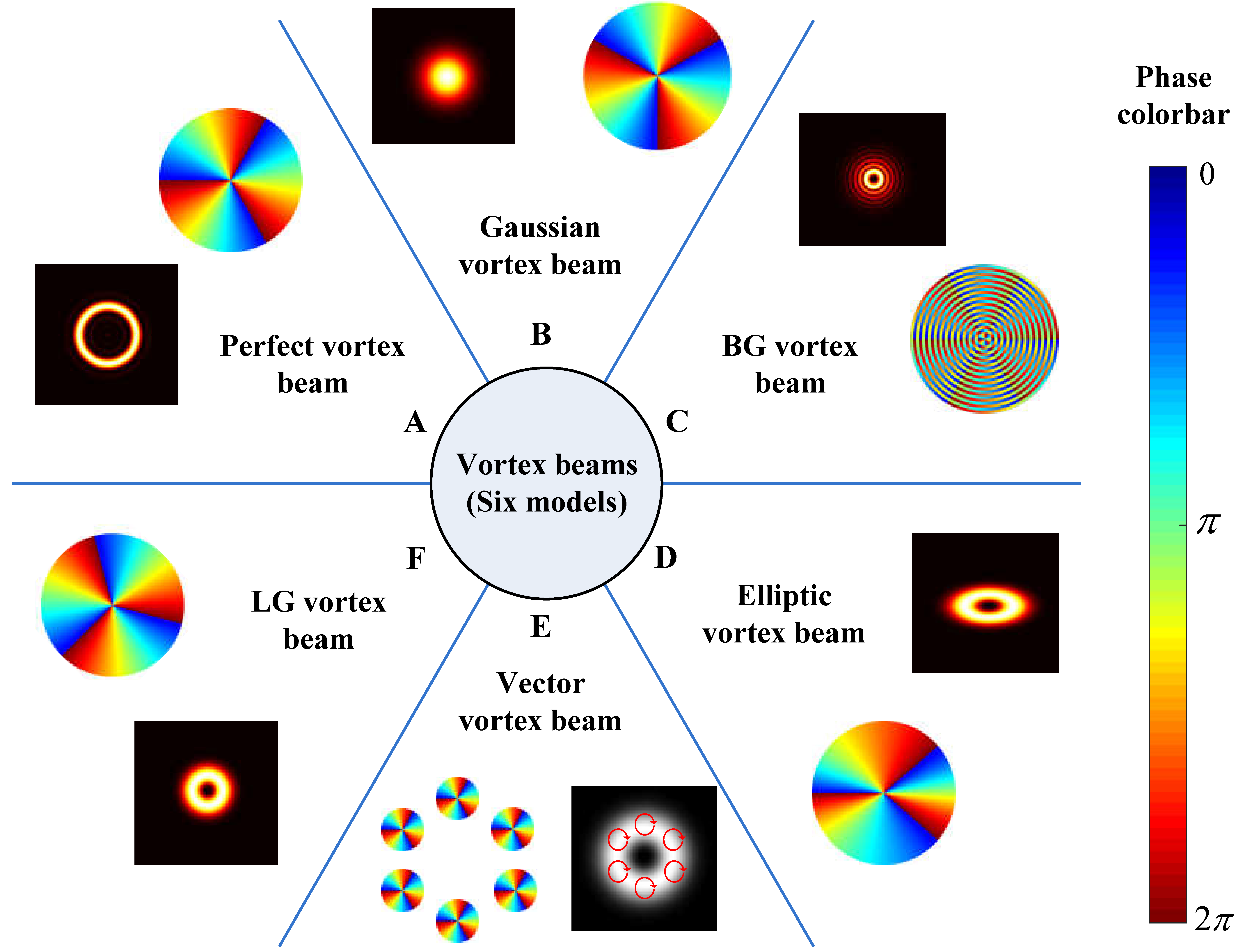}}
\caption{Intensity and phase patterns of the theoretical models of six categories of vortex beams with topological charge (TC) set as 3.  }
\label{Intensity_phase}
\end{figure}

Based on the intrinsic characteristics, OAM beams can be classified into perfect vortex beams, Gaussian vortex beams, Bessel–Gaussian (BG) vortex beams, elliptic vortex beams, vector vortex beams, and Laguerre–Gaussian (LG) vortex beams. The intensity and phase patterns of the above six categories of vortex beams are illustrated in Fig. \ref{Intensity_phase} and the differences among these beam modes are shown in Table \ref{tbl:comp1}, wherein BG vortex beams and LG vortex beams are widely used in electromagnetic field, such as wireless communication and radar application. In this paper, we adopt the LG beam scheme rather than BG beam scheme to describe OAM beams since the former one can provide lower generation cost and a simpler way for subsequent derivation and analysis. In a cylindrical coordinate system, the LG beam is exhibited as \cite{XGe2017}
\begin{subequations}
\begin{align}           \label{eq:Prime_1}
\!\!\!  \!\!\!  \mathcal{A}(r,\varphi,z) \!\! &=\! \gamma \sqrt{\frac{\mathfrak{P}!}{\pi(\mathfrak{P}+\vert l \vert)!}} \frac{1}{w_l(z)} \left( \frac{\sqrt{2}r}{w_l(z)} \right)^{\vert l \vert} e^{-\left(\frac{r}{w_l(z)}\right)^2}    \notag \\
    & \times \! \mathbb{L}_{\mathfrak{P}}^{\vert l \vert} \left( \frac{2r^2}{w_l^2(z)} \right) \! e^{-\mathrm{j}\frac{\pi r^2}{\lambda R_l(z)}} e^{\mathrm{j}(\vert l \vert + 2\mathfrak{P}+1 )\zeta(z)} e^{-\mathrm{j}l\varphi},
\end{align}
with
\begin{align}
    w_l(z) = w_l \sqrt{1+\left( \frac{z}{z_R} \right)^2},~~ R_l(z)=z\left[ 1+\left( \frac{\pi w_l^2}{\lambda z} \right)^2 \right],
\end{align}
\end{subequations}
where  $(r,\varphi, z)$ is the cylindrical coordinate with radial distance $r$, azimuthal angle $\varphi$ and propagation distance $z$, respectively; $\gamma \sqrt{\frac{\mathfrak{P}!}{\pi(\mathfrak{P}+\vert l \vert)!}}$ is a normalized constant with $\mathfrak{P}$ being the radial index and $l$ being the value of OAM state; $w_l$ is the beam waist radius corresponding to the OAM state $l$ with $z = 0$; $z_R= \pi w_l^2/\lambda$ is the Rayleigh distance where $\lambda$ denotes the wavelength; $ \mathbb{L}_{\mathfrak{P}}^{\vert l \vert} \left( \frac{2r^2}{w_l^2(z)} \right)$ is the generalized Laguerre polynomial; and $\zeta(z)$ is the Gouy phase which stems from the phase velocity of signal. Since the transverse intensity distribution of an LG vortex beam exhibits a singularity at the center, the energy of OAM signal is focused in a ring, namely OAM circle region. The radius of the OAM circle region with the maximum energy strength can be configured by 
\begin{align}
    r_{\mathrm{max}}(z) =w_l \sqrt{\frac{\vert l \vert (1+ ( z/z_R )^2 )}{2}}.
\end{align}
In our proposed OAM-based JRC system, we set $\mathfrak{P}=0$ leading to $\mathbb{L}_{\mathfrak{P}}^{\vert l \vert} \left( \frac{2r^2}{w_l^2(z)} \right)=1$.

Fig. \ref{OAM_beam} displays the wavefronts and azimuthally varying phases for different OAM numbers. For conventional planar waveforms, the OAM state number $l$ is equal to 0, where the wavefronts experience no contortion. For other OAM state numbers which are not equal to 0, the wavefronts exhibit helical phase profile. Besides that, another difference between OAM wave and planar wave lies in beam intensity distribution. The zeroth-order beam, namely planar wave beam, lacks an azimuthal phase variation so that it does not exhibit a singularity and the intensity has a Gaussian shape. For a higher-order beam, namely OAM beam, a point of undefined phase arises due to the helical phase evolution around the propagation axis, which manifests itself as a central singularity. The intensity distribution of OAM beam is a function of radial distance and OAM state number for a fixed $z$-plane. 

\begin{figure}[t]
\centerline{\includegraphics[scale=0.46]{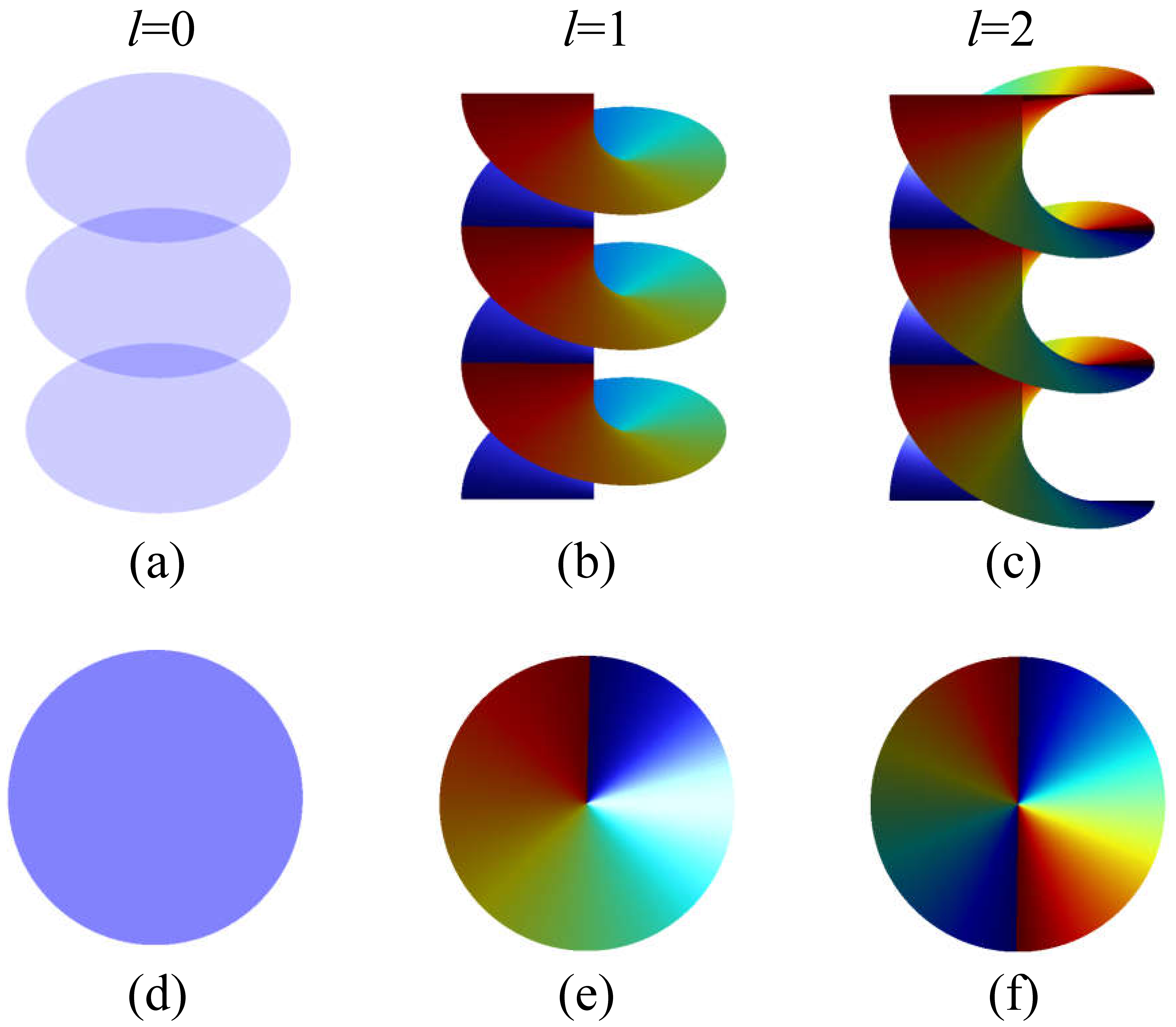}}
\caption{Depiction of the (a)-(c) wavefronts for OAM numbers 0, 1 and 2, respectively and (d)-(f) are the corresponding azimuthally varying phase plates of vortex beams.  }
\label{OAM_beam}
\end{figure}

\subsection{OAM-JRC}  
\begin{figure}[t]
\centerline{\includegraphics[scale=0.32]{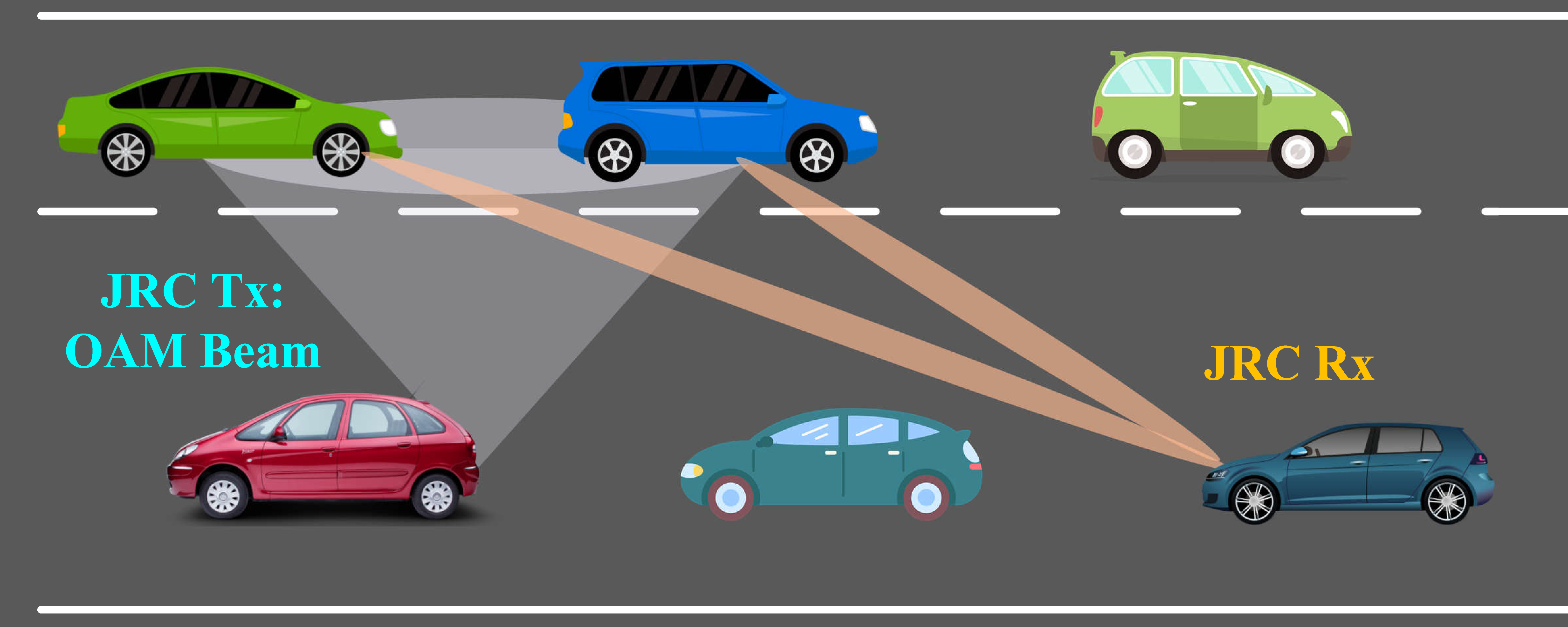}}
\caption{A bi-static OAM-JRC scenario: a common transmitter (Tx) on a vehicle (bottom lane, left) sends out a JRC signal that reflects off other vehicles (top lane) and is intercepted by a receiver (Rx) on another vehicle (bottom lane, right). }
\label{Scenario}
\end{figure}
Consider the automotive driving scenario on a normal road depicted in Fig. \ref{Scenario}. The JRC transmit (receive) antenna is a frequency-diverse-array structure that sends out (receives) vortex EM wave carrying OAM and comprises $M$ ($N$) elements in a ULA whose interelement spacing is denoted by $d$. The transmit array elements are equipped with travelling-wave ring resonators which can be used to generate LG vortex beams with different OAM states simultaneously \cite{XXiong2022_1}. For ease of analysis, different OAM beams are assumed to have approximate sizes of circle regions, denoted by $r_{\mathrm{max}}(z)$, for the same $z$. The frequency diverse array (FDA) employs a unit frequency increment $\Delta f$ between adjacent antennas so that the radiation frequency of the $m$-th transmit element takes the form
\begin{align}
    f_m=f_0-m \Delta f,~~~ m=0,1,\ldots,M-1
\end{align}
where $f_0$ is the reference carrier frequency.

\begin{figure*}[t]
\centerline{\includegraphics[scale=0.68]{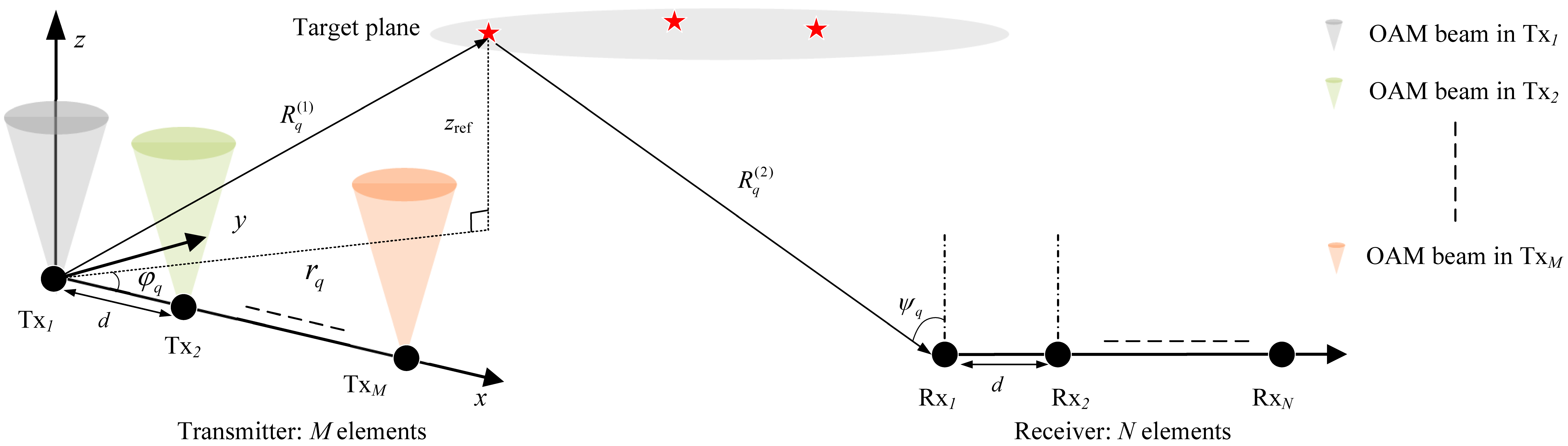}}
\caption{The simplified schematic diagram of Fig. \ref{Scenario}. }
\label{Schematic}
\end{figure*}

For a clear display, we predigest Fig. \ref{Scenario} to the schematic diagram in Fig. \ref{Schematic}. The first transmit element is regarded as the origin of the cylindrical coordinate system of the transmitting ULA. The position of the $m$-th transmit element is denoted as $\bm{\xi}_m=(md,0,0)$. The $x$-$y$ plane is perpendicular to the ground where $x$-axis is parallel to the white lane dividing line and $y$-axis is perpendicular to the ground. $z$-axis is perpendicular to the $x$-$y$ plane. The scattering surfaces of different automobiles (targets) can be assumed to be in the same plane, namely target plane, which is parallel to the $x$-$y$ plane. Each transmit element generates OAM beams with $2U+1$ different states whose state number is characterized by the set $\mathcal{M}=\{l_{-U},\ldots,l_{-1},l_0,l_1,\ldots,l_U \}$ where $l_{-u} \triangleq -l_u$, $-U\leq u \leq U$, and $l_0=0$. Here, we take a uniform pattern of $\mathcal{M}$ into account, whose inter-state number spacing is $\delta$, namely $l_u=u\delta$. The continuous-time signal transmitted from the $m$-th element at the $u$-th OAM state is
\begin{align}                     \label{eq:SM_2}
    s_m^{l_u}(t) = A_m a_u e^{\mathrm{j}2\pi f_m t} ,
\end{align}
where $a_u$ denotes the complex data symbol with duration $T_{\mathrm{sym}}$ at the $u$-th OAM state and $A_m$ is the amplitude. Without loss of generality, we set $A_m=1$.

In automotive scenario, we take extended targets into account. Basically, an extended target can be modeled as a cluster of independent point scatterers. Therefore, assume that there are totally $Q$ point scatterers from all extended targets and $\tau_q^{(1)}=R_q^{(1)}/c$ is the time delay between the $q$-th, $1\leq q \leq Q$, point scatterer and transmitting array where $R_q^{(1)}$ is the corresponding distance. Denote $z_{\mathrm{ref}}$ as the distance between the origin and the target plane, $d_{\mathrm{ref}}$ as the distance between the origin and the OAM circle region which has a radius $r_{\mathrm{max}}(z_{\mathrm{ref}})$ in the target plane. In a narrowband flat fading channel, the instantaneous phase at the $q$-th scatterer is the sum of the distance-dependent phase and the azimuthal phase. Thus, by adopting the OAM-based wireless channel response model which is derived from the transmitted LG beams $\mathcal{A}(r,\varphi,z)$ in (\ref{eq:Prime_1}) at each transmit element \cite{XXiong2022}, the signal at the $u$-th OAM state from the $m$-th transmit element impinging on the $q$-th scatterer can be written as
\begin{align}               \label{eq:SM_2_1}
   z_{m,q}^{l_u} (t) &= \frac{\lambda_m}{4\pi d_{\mathrm{ref}}} \left( \frac{r_{mq}}{r_{\mathrm{max}}(z_{\mathrm{ref}})} \right)^{\vert l_u \vert} e^{-\frac{r_{mq}^2- r^2_{\mathrm{max}}(z_{\mathrm{ref}})}{w^2_{l_u}(z_{\mathrm{ref}})}}                           \notag  \\
   &~~~ \times e^{-\mathrm{j}\pi \frac{r_{mq}^2- r^2_{\mathrm{max}}(z_{\mathrm{ref}})}{\lambda_m R_{l_u}(z_{\mathrm{ref}})}} e^{-\mathrm{j}2\pi \frac{d_{\mathrm{ref}}}{\lambda_m} } e^{-\mathrm{j}l_u \varphi_q }  a_u  e^{\mathrm{j}2\pi f_m (t-\tau_q^{(1)})}               \notag \\
   &\approx \frac{\lambda_m}{4\pi d_{\mathrm{ref}}} \left( \frac{r_{q}}{r_{\mathrm{max}}(z_{\mathrm{ref}})} \right)^{\vert l_u \vert} e^{-\frac{r_{q}^2- r^2_{\mathrm{max}}(z_{\mathrm{ref}})}{w^2_{l_u}(z_{\mathrm{ref}})}}                           \notag  \\
   &~~~ \times e^{-\mathrm{j}\pi \frac{r_{q}^2- r^2_{\mathrm{max}}(z_{\mathrm{ref}})}{\lambda_m R_{l_u}(z_{\mathrm{ref}})}} e^{-\mathrm{j}2\pi \frac{d_{\mathrm{ref}}}{\lambda_m} } e^{-\mathrm{j}l_u \varphi_q }  a_u  e^{\mathrm{j}2\pi f_m (t-\tau_q^{(1)})} 
\end{align}
where $\lambda_m = c/f_m$, $c=3\times 10^8 \mathrm{m/s}$ is the velocity of electromagnetic wave propagation, $\varphi_q$ is the azimuthal angle between the $q$-th scatterer and transmitting array and $r_{mq}$ is the radial distance in the cylindrical coordinate system between the $m$-th transmit element and the $q$-th scatterer. The approximation in (\ref{eq:SM_2_1}) stems from the far-field assumption so that $r_{mq} \approx r_q$. Obviously, the radial distance $r_q$ is related with the transmitter-scatterer distance $R_q^{(1)}$ by 
\begin{align}
    R_q^{(1)} = \sqrt{r_q^2+ z_{\mathrm{ref}}^2}.
\end{align}

The total flight time corresponding to a bi-static range $R_q = R_q^{(1)}+R_q^{(2)}$, where the superscript (2) denotes variable dependency on the target-receiver path, is $\tau_q=\tau_q^{(1)}+\tau_q^{(2)}$. Assume $\psi_q$ to be the angle between the $q$-th scatterer echo and the normal line of receiving ULA. Then, the $m$-th frequency component with OAM state $l_u$ reflected by targets and received by the $n$-th, $0\leq n \leq N-1$, receive element takes the form 
\begin{align}          \label{eq:SM_3}
    x_{m,n}^{l_u}(t) &= \sum_{q=1}^Q  \rho_q(t) z_{m,q}^{l_u} (t-\tau_q^{(2)}) e^{-\mathrm{j}2\pi d\sin(\psi_q)n/\lambda_m} + n_{m,n}^{l_u}(t)      \notag \\
    &=  \sum_{q=1}^Q  \rho_q(t) a_u \frac{\lambda_m}{4\pi d_{\mathrm{ref}}} \left( \frac{r_{q}}{r_{\mathrm{max}}(z_{\mathrm{ref}})} \right)^{\vert l_u \vert} e^{-\frac{r_{q}^2- r^2_{\mathrm{max}}(z_{\mathrm{ref}})}{w^2_{l_u}(z_{\mathrm{ref}})}}  \notag  \\
   &~~~ \times e^{-\mathrm{j}\pi \frac{r_{q}^2- r^2_{\mathrm{max}}(z_{\mathrm{ref}})}{\lambda_m R_{l_u}(z_{\mathrm{ref}})}} e^{-\mathrm{j}2\pi \frac{d_{\mathrm{ref}}}{\lambda_m} } e^{-\mathrm{j}l_u \varphi_q }    e^{\mathrm{j}2\pi f_m (t-\tau_q)}                       \notag \\
   &~~~ \times e^{-\mathrm{j}2\pi d \sin(\psi_q)n/\lambda_m} + n_{m,n}^{l_u}(t)           \notag \\
   &= \sum_{q=1}^Q  \acute{\rho}_q(t) a_u  e^{-\mathrm{j}\pi \frac{r_{q}^2- r^2_{\mathrm{max}}(z_{\mathrm{ref}})}{\lambda_m R_{l_u}(z_{\mathrm{ref}})}} e^{-\mathrm{j}2\pi \frac{d_{\mathrm{ref}}}{\lambda_m} } e^{-\mathrm{j}l_u \varphi_q }        \notag \\
   &~~~ \times e^{\mathrm{j}2\pi f_m (t-\tau_q)} e^{-\mathrm{j}2\pi d\sin(\psi_q)n/\lambda_m} + n_{m,n}^{l_u}(t) ,
\end{align}
where $n_{m,n}^{l_u}(t)$ is the additive spatially and temporally white noise with the power $\sigma^2_n$, $\{\rho_q(t)\}_{q=1}^Q$ are the complex scattering coefficients of scatterers, satisfying $\mathrm{E}[\rho_p^* \rho_q] = \sigma_q^2 \delta_{p,q}$. In our model of (\ref{eq:SM_3}), amplitude terms have no effect on parameter estimation where we can ignore RCS dependency on the transmit elements and OAM states, thus we have
\begin{align}
    \acute{\rho}_q(t) = \rho_q(t) \frac{\lambda_m}{4\pi d_{\mathrm{ref}}} \left( \frac{r_{q}}{r_{\mathrm{max}}(z_{\mathrm{ref}})} \right)^{\vert l_u \vert} e^{-\frac{r_{q}^2- r^2_{\mathrm{max}}(z_{\mathrm{ref}})}{w^2_{l_u}(z_{\mathrm{ref}})}} .
\end{align}
For the phase term $e^{-\mathrm{j}2\pi \frac{d_{\mathrm{ref}}}{\lambda_m} }$, it only depends on the transmit antenna which can be compensated directly. Thus, by compensation, demodulation and applying band-pass filtering, the baseband signal corresponding to (\ref{eq:SM_3}) can be written as
\begin{align}          \label{eq:SM_4}
   x_{m,n}^{l_u}(t_s) 
   &=\sum_{q=1}^Q  \acute{\rho}_q(t_s) a_u  e^{-\mathrm{j}\pi \frac{r_{q}^2- r^2_{\mathrm{max}}(z_{\mathrm{ref}})}{\lambda_m R_{l_u}(z_{\mathrm{ref}})}} e^{-\mathrm{j}l_u \varphi_q }       \notag \\
   &~~~ \times e^{-\mathrm{j}2\pi (f_0 - m \Delta f) \tau_q} e^{-\mathrm{j}2 \pi d \sin(\psi_q)n/{\lambda_m}} + n_{m,n}^{l_u}(t_s)      \notag \\
   &\approx \sum_{q=1}^Q a_u \widetilde{\rho}_q(t_s)  e^{\mathrm{j}2\pi m \Delta f R_q/c} e^{-\mathrm{j}\left[  \pi \frac{r_{q}^2- r^2_{\mathrm{max}}(z_{\mathrm{ref}})}{\lambda_0 R_{l_u}(z_{\mathrm{ref}})} + l_u \varphi_q \right]} \notag \\
        & ~~~ \times e^{-\mathrm{j}2 \pi d \sin(\psi_q)n f_0/c} + n_{m,n}^{l_u}(t_s),  ~ t_s \in [0,T_{\mathrm{sym}}),
\end{align}
where $\widetilde{\rho}_q(t_s) = \acute{\rho}_q(t_s) e^{-\mathrm{j}2\pi f_0 \tau_q}$. The approximation originates from the narrowband assumption that $M\Delta f \ll f_0$. For convenience of representation, denote 
\begin{align}        \label{eq:SM_5}
    \Phi_q^{l_u} \triangleq  -\pi \frac{r_{q}^2- r^2_{\mathrm{max}}(z_{\mathrm{ref}})}{\lambda_0 R_{l_u}(z_{\mathrm{ref}})} - l_u \varphi_q.
\end{align}
Then, given an OAM state, say $l_u$, collecting different frequency diverse signals across different receive elements leads to the following slice of data vector
\begin{align}         \label{eq:SM_6}
    \mathbf{x}_{l_u}(t_s) &=   \sum_{q=1}^Q  \widetilde{\rho}_q(t_s) a_u e^{\mathrm{j} \Phi_q^{l_u}}   \mathbf{b}(R_q) \otimes  \mathbf{a}_R(\psi_q)  +   \mathbf{n}_{l_u}(t_s),
\end{align}
where $\mathbf{b}(R_q)$ is an $M\times 1$ vector with the $m$-th element as $e^{\mathrm{j}2\pi m\Delta f R_q/c}$ and $\mathbf{a}_R(\psi_q)$ is an $N\times 1$ vector with the $n$-th element as $e^{-\mathrm{j}2 \pi d \sin(\psi_q)n f_0/c}$.

In the proposed OAM-based JRC system, our goal is to recover the unknown scatterer parameters: direction of departure (DoD), namely the azimuthal angle $\varphi_q$, DoA $\psi_q$, bi-static range $R_q$, radial distance $r_q$, moving velocity and communication symbols $a_u$. We make the following assumptions about the target and radar parameters:
\begin{description}
  
    \item[A1] ``Unambiguous DoA and bi-static range'': To ensure the array structure deprived of ambiguity, assume
\begin{align}                  \label{eq:JRC_A3}
   \sin \psi_q \neq \sin \psi_p,~~~R_q \neq R_p ~~~\text{for}~~~ 1\leq p & \neq  q \leq Q.
\end{align}
    In addition, the bi-static range parameters $\{R_q \}_{q=1}^Q$ are assumed to be located in the interval $[z_{\mathrm{ref}},R_{\max}]$ where $R_{\mathrm{max}}=\frac{c}{\Delta f}$ is the maximum unambiguous range.
    
    \item[A2]  ``No DoD or radial distance ambiguities'': The DoD parameters $\{\varphi_q\}_{q=1}^Q$ lie in the interval $[\varphi_{\min},\varphi_{\max}]$ where $\varphi_{\min}=-\frac{\pi}{2\delta}$ and $\varphi_{\max}=\frac{\pi}{2\delta}$. The radial distance parameters $\{r_q \}_{q=1}^Q$ lie in the interval $[\tilde{r}_{\min},\tilde{r}_{\max}]$ where 
    \begin{align}
        \tilde{r}_{\min} = \sqrt{r_{\mathrm{max}}^2(z_{\mathrm{ref}})-\frac{\lambda_0}{2} \min_{l_u}R_{l_u}(z_{\mathrm{ref}})},  \\
        \tilde{r}_{\max} = \sqrt{r_{\mathrm{max}}^2(z_{\mathrm{ref}})+\frac{\lambda_0}{2} \min_{l_u}R_{l_u}(z_{\mathrm{ref}})}.
    \end{align}

\end{description}

\section{OAM-JRC Receive Processing}     \label{sec:algorithm}
Here, we adopt the OAM-based MDM between the radar and JRC frames. Assume that $\mu$ percent of the OAM carriers are allocated to radar and the rest to JRC. $\mu$ can be regarded as the sharing factor in JRC. 

\subsection{OAM-based Target Position Parameters Recovery}        \label{JRC_Radar}
To obtain the radar-only signal model, we set the symbols $a_u$ in (\ref{eq:SM_4}) to be ones in the communication-free frames. Then, by stacking all the slices of data vector with different OAM states in (\ref{eq:SM_6}), we have
\begin{align}              \label{eq:JRC_A1}
    \mathbf{x}_r(t_s)\! \!= \! \sum_{q=1}^Q  \widetilde{\rho}_q(t_s) \mathbf{a}_{Tr}(r_q,\varphi_q) \otimes  \mathbf{b}(R_q) \otimes  \mathbf{a}_R(\psi_q) \!+\! \mathbf{n}_r(t_s),
\end{align}
where 
\begin{align}
    \mathbf{a}_{Tr}(r_q,\varphi_q) &= \left[e^{\mathrm{j}\Phi_q^{l_{-\mu U}}},\ldots,e^{\mathrm{j}\Phi_q^{l_{-1}}},e^{\mathrm{j}\Phi_q^{l_{0}}},e^{\mathrm{j}\Phi_q^{l_{1}}}\ldots,e^{\mathrm{j}\Phi_q^{l_{\mu U}}}\right]^T \notag \\
    &~~~\in \mathbb{C}^{(2\mu U+1) \times 1}.
\end{align}

The $(2\mu U +1) MN\times (2\mu U+1) MN$ covariance matrix of data vector $\mathbf{x}_r(t_s)$ is obtained as
\begin{align}              \label{eq:JRC_A2}
\!\!\!     \mathbf{R}_r \! \!&= \! \mathrm{E}\{ \mathbf{x}_r(t_s) \mathbf{x}_r^H(t_s)  \}    \notag \\
    &= \! (\mathbf{A}_{Tr} \odot \mathbf{B} \odot\mathbf{A}_R) \mathbf{R}_{\rho}  (\mathbf{A}_{Tr} \odot \mathbf{B} \odot\mathbf{A}_R)^H \! \!+ \! \sigma_n^2 \mathbf{I}_{(2\mu U+1) MN},
\end{align}
where $\mathbf{A}_{Tr}=[\mathbf{a}_{Tr}(r_1,\varphi_1),\mathbf{a}_{Tr}(r_2,\varphi_2),\ldots,\mathbf{a}_{Tr}(r_Q,\varphi_Q)] \in \mathbb{C}^{(2\mu U+1) \times Q}$, $\mathbf{B}=[\mathbf{b}(R_1),\mathbf{b}(R_2),\ldots,\mathbf{b}(R_Q)]\in \mathbb{C}^{M\times Q}$, $\mathbf{A}_R=[\mathbf{a}_R(\psi_1),\mathbf{a}_R(\psi_2),\ldots,\mathbf{a}_R(\psi_Q)] \in  \mathbb{C}^{N\times Q}$. $\mathbf{R}_{\rho}$ is a diagonal matrix with the entries $\{\sigma_1^2,\sigma_2^2,\ldots,\sigma_Q^2 \}$. 

Based on the assumption that the signals and noise are independent with each other, the covariance matrix $\mathbf{R}_r$ can be decomposed into two mutually orthogonal parts:
\begin{align}       \label{eq:JRC_A3_1}
    \mathbf{R}_r &=  [\mathbf{U}_{rs} ~\mathbf{U}_{rn}] 
    \left[
     \begin{matrix}
        \bm{\Lambda}_{rs} & \mathbf{0}  \\
        \mathbf{0}   &   \bm{\Lambda}_{rn}
     \end{matrix}
    \right] 
    [\mathbf{U}_{rs} ~\mathbf{U}_{rn}]^H       \notag \\
    &= \mathbf{U}_{rs} \bm{\Lambda}_{rs} \mathbf{U}_{rs}^H + \mathbf{U}_{rn}  \bm{\Lambda}_{rn} \mathbf{U}_{rn}^H,
\end{align}
where $\bm{\Lambda}_{rs}$ is the $Q$-dimensional diagonal matrix containing the larger $Q$ eigenvalues of $ \mathbf{R}_r$ and $\bm{\Lambda}_{rn}$ is the $[(2\mu U+1) MN - Q]$-dimensional diagonal matrix containing the smaller $[(2\mu U+1) MN - Q]$ eigenvalues. The columns of the matrix $ [\mathbf{U}_{rs} ~\mathbf{U}_{rn}]$ are the eigenvectors of $\mathbf{R}_r$, where $\mathbf{U}_{rs}$ is the signal subspace containing the vectors corresponding to the larger $Q$ eigenvalues and $\mathbf{U}_{rn}$ is the noise subspace containing the vectors corresponding to the smaller $[(2\mu U+1) MN - Q]$ eigenvalues. Assume that there exists an invertible $Q\times Q$ matrix $\mathbf{T}$ such that 
\begin{align}             \label{eq:JRC_A2_2}
    \mathbf{U}_{rs}=(\mathbf{A}_{Tr} \odot \mathbf{B} \odot\mathbf{A}_R)\mathbf{T}.
\end{align}
The sufficient conditions for such a $\mathbf{T}$ to exist are included in Theorem \ref{Theo:RG} of Section \ref{sec:PA}.

Denote the first $(\mu U +1)$ rows of $\mathbf{A}_{Tr}$, the first $(M-1)$ rows of $\mathbf{B}$, and the first $(N-1)$ rows of $\mathbf{A}_R$ as $\mathbf{A}_{Tr1}$, $\mathbf{B}_1$, and $\mathbf{A}_{R1}$, respectively. The last $(\mu U + 1)$ rows of $\mathbf{A}_{Tr}$, the last $(M-1)$ rows of $\mathbf{B}$, and the last $(N-1)$ rows of $\mathbf{A}_R$ are denoted as $\mathbf{A}_{Tr2}$, $\mathbf{B}_2$, and $\mathbf{A}_{R2}$, respectively. 

For the matrices $\mathbf{B}$ and $\mathbf{A}_R$, each couple of submatrices along the same dimension are related as $\mathbf{B}_2=\mathbf{B}_1 \bm{\Psi}$ and $\mathbf{A}_{R2}=\mathbf{A}_{R1} \bm{\Omega}$ where $\bm{\Psi}=\mathrm{diag}([e^{\mathrm{j}2\pi \Delta f R_1/c},e^{\mathrm{j}2\pi \Delta f R_2/c},\ldots,e^{\mathrm{j}2\pi \Delta f R_Q/c}])$ and $\bm{\Omega}=\mathrm{diag}([e^{-\mathrm{j}2\pi d\sin(\psi_1)f_0/c},e^{-\mathrm{j}2\pi d\sin(\psi_2)f_0/c},\ldots,e^{-\mathrm{j}2\pi d\sin(\psi_Q)f_0/c}])$. By selecting the specific rows from $\mathbf{U}_{rs}$, we can construct
\begin{align}                  \label{eq:JRC_A4}
         \mathbf{U}_{B_1} &=  (\mathbf{A}_{Tr} \odot \mathbf{B}_1 \odot\mathbf{A}_R)\mathbf{T},     \notag \\
         \mathbf{U}_{B_2} &=  (\mathbf{A}_{Tr} \odot \mathbf{B}_2 \odot\mathbf{A}_R) \mathbf{T}   \notag \\  
         &= (\mathbf{A}_{Tr} \odot \mathbf{B}_1 \odot\mathbf{A}_R) \bm{\Psi }\mathbf{T}   \notag \\  
         &=  \mathbf{U}_{B_1}   \mathbf{T}^{-1}  \bm{\Psi} \mathbf{T}.
\end{align}
Similarly, we have
\begin{align}                  \label{eq:JRC_A5}
         \mathbf{U}_{A_{R1}} &=  (\mathbf{A}_{Tr} \odot \mathbf{B} \odot\mathbf{A}_{R1})\mathbf{T},   \notag \\
         \mathbf{U}_{A_{R2}} &=  (\mathbf{A}_{Tr} \odot \mathbf{B} \odot\mathbf{A}_{R2})\mathbf{T}    \notag \\
         &=(\mathbf{A}_{Tr} \odot \mathbf{B} \odot\mathbf{A}_{R1}) \bm{\Omega} \mathbf{T}       \notag \\
         &= \mathbf{U}_{A_{R1}} \mathbf{T}^{-1} \bm{\Omega}\mathbf{T}.
\end{align}
Then, based on (\ref{eq:JRC_A4}) and (\ref{eq:JRC_A5}), we have
\begin{align}
     \mathbf{U}_{B_1}^{\dag}  \mathbf{U}_{B_2} &=  \mathbf{T}^{-1}  \bm{\Psi} \mathbf{T},          \label{eq:JRC_A6_1}   \\
     \mathbf{U}_{A_{R1}}^{\dag} \mathbf{U}_{A_{R2}} &=  \mathbf{T}^{-1}  \bm{\Omega} \mathbf{T}  .  \label{eq:JRC_A6_2} 
\end{align}
Obviously, $\mathbf{T}$ and $\bm{\Psi}$ can be obtained using the eigenvalue decomposition of $ \mathbf{U}_{B_1}^{\dag}  \mathbf{U}_{B_2}$, up to permutation. Denote the resulting matrices by $\widehat{\bm{\Psi}}$ and $\widehat{\mathbf{T}}$. Compute $\widehat{\bm{\Omega}}$ as $\widehat{\bm{\Omega}}= \widehat{\mathbf{T}} (\mathbf{U}_{A_{R1}}^{\dag} \mathbf{U}_{A_{R2}}) \widehat{\mathbf{T}}^{-1}  $. Then, the recovered bi-static ranges and DoAs are
\begin{align}                 \label{eq:JRC_A7}
\widehat{R}_q = \frac{\angle [\widehat{\bm{\Psi}}]_{q,q}}{2\pi \Delta f/c}, ~~  \widehat{\psi}_q = -\arcsin\left(\frac{\angle [\widehat{\bm{\Omega}}]_{q,q}}{2\pi df_0/c} \right).
\end{align}

For the matrix $\mathbf{A}_{Tr}$, consider the phase term $\Phi_q^{l_u}$ which is defined in (\ref{eq:SM_5}). Because the radial distance $r_q$ and azimuthal angle $\varphi_q$ are coupled with each other, a new decoupling method is explored here. Based on the fact that $R_{l_u}(z_{\mathrm{ref}})=R_{-l_u}(z_{\mathrm{ref}})$ \cite{XGe2017}, we have
\begin{align}
    \Phi_q^{l_u} + \Phi_q^{-l_u} &= -2\pi \frac{r_{q}^2- r^2_{\mathrm{max}}(z_{\mathrm{ref}})}{\lambda_0 R_{l_u}(z_{\mathrm{ref}})}    \notag  \\ 
                                 &= -2\pi \frac{r_{q}^2- r^2_{\mathrm{max}}(z_{\mathrm{ref}})}{\lambda_0 R_{u\delta }(z_{\mathrm{ref}})}  ,               \\
    \Phi_q^{l_u} - \Phi_q^{-l_u} &= -2l_u \varphi_q= -2 u\delta \varphi_q .
\end{align}
Denote $\widetilde{\mathbf{A}}_{Tr1} \triangleq \mathrm{flipud}(\mathbf{A}_{Tr1})$. By selecting the specific rows from $\mathbf{U}_{rs}$, we can construct 
\begin{align}             \label{eq:JRC_A9} 
    \mathbf{U}_{A_{Tr1}} &=  (\widetilde{\mathbf{A}}_{Tr1} \odot \mathbf{B} \odot\mathbf{A}_{R})\mathbf{T},        \notag   \\
    \mathbf{U}_{A_{Tr2}} &=  (\mathbf{A}_{Tr2} \odot \mathbf{B} \odot\mathbf{A}_{R})\mathbf{T}.                 
\end{align}
Applying the inverse of $\mathbf{T}$ on the left of (\ref{eq:JRC_A9}) and multiplying the two matrices element by element, we can get
\begin{align}              \label{eq:JRC_A10}
& ~  (\mathbf{U}_{A_{Tr1}} \mathbf{T}^{-1}) \diamond (\mathbf{U}_{A_{Tr2}} \mathbf{T}^{-1})              \notag \\
&= (\widetilde{\mathbf{A}}_{Tr1} \odot \mathbf{B} \odot\mathbf{A}_{R}) \diamond (\mathbf{A}_{Tr2} \odot \mathbf{B} \odot\mathbf{A}_{R})      \notag \\
&= (\widetilde{\mathbf{A}}_{Tr1} \diamond \mathbf{A}_{Tr2}) \odot (\mathbf{B} \diamond \mathbf{B}) \odot (\mathbf{A}_{R} \diamond \mathbf{A}_{R}).
\end{align}
Similarly, we construct 
\begin{align}             \label{eq:JRC_A11}
    & ~  (\mathbf{U}_{A_{Tr1}} \mathbf{T}^{-1}) \diamond (\mathbf{U}_{A_{Tr2}} \mathbf{T}^{-1})^*             \notag \\
&= (\widetilde{\mathbf{A}}_{Tr1} \odot \mathbf{B} \odot\mathbf{A}_{R}) \diamond (\mathbf{A}_{Tr2}^* \odot \mathbf{B}^* \odot\mathbf{A}_{R}^*)      \notag \\
&= (\widetilde{\mathbf{A}}_{Tr1} \diamond \mathbf{A}_{Tr2}^*) \odot (\mathbf{B} \diamond \mathbf{B}^*) \odot (\mathbf{A}_{R} \diamond \mathbf{A}_{R}^*)    \notag \\
&= (\widetilde{\mathbf{A}}_{Tr1} \diamond \mathbf{A}_{Tr2}^*) \odot \mathbf{1}_{M\times Q} \odot \mathbf{1}_{N\times Q}.
\end{align}
By substituting the estimated $\{\widehat{R}_q\}_{q=1}^Q$, $\{\widehat{\psi}_q\}_{q=1}^Q$ and $\widehat{\mathbf{T}}$ into (\ref{eq:JRC_A10}) and (\ref{eq:JRC_A11}), the matrices $\mathbf{A}^{+} \triangleq \widetilde{\mathbf{A}}_{Tr1} \diamond \mathbf{A}_{Tr2}$ and $\mathbf{A}^{-} \triangleq \widetilde{\mathbf{A}}_{Tr1} \diamond \mathbf{A}_{Tr2}^*$ can be recovered. The $q$-th column of $\mathbf{A}^{+}$ and $\mathbf{A}^{-}$, which corresponds to the $q$-th target, takes the form
\begin{align}
    [\mathbf{A}^{+}]_q = \left[
       \begin{matrix}
           \mathrm{exp}\left( -\mathrm{j}2\pi \frac{r_{q}^2- r^2_{\mathrm{max}}(z_{\mathrm{ref}})}{\lambda_0 R_{0}(z_{\mathrm{ref}})} \right)   \\
           \mathrm{exp}\left( -\mathrm{j}2\pi \frac{r_{q}^2- r^2_{\mathrm{max}}(z_{\mathrm{ref}})}{\lambda_0 R_{\delta }(z_{\mathrm{ref}})} \right)  \\
           \vdots   \\
           \mathrm{exp}\left( -\mathrm{j}2\pi \frac{r_{q}^2- r^2_{\mathrm{max}}(z_{\mathrm{ref}})}{\lambda_0 R_{\mu U \delta }(z_{\mathrm{ref}})} \right)
       \end{matrix}
    \right]   \in \mathbb{C}^{(\mu U+1)\times 1}
\end{align}
and
\begin{align}
    [\mathbf{A}^{-}]_q = [1, e^{ -\mathrm{j}2 \delta \varphi_q}, \ldots,e^{ -\mathrm{j}2 \mu U \delta \varphi_q}] \in \mathbb{C}^{(\mu U+1)\times 1},
\end{align}
respectively. Then, the recovered radial distance and azimuthal angle of the $q$-th scatterer can be calculated as
\begin{align}         \label{eq:JRC_A14}   
    \widehat{r}_q = \frac{1}{\mu U+1}\sum_{i=1}^{\mu U+1} \sqrt{\frac{-\angle[\mathbf{A}^{+}]_{i,q}\lambda_0 R_{(i-1)\delta}(z_{\mathrm{ref}})}{2\pi}+ r^2_{\mathrm{max}}(z_{\mathrm{ref}})}         
\end{align}
and
\begin{align}   \label{eq:JRC_A15}
    \widehat{\varphi}_q =  \frac{1}{2\mu U \delta} \sum_{i=2}^{\mu U+1} \angle\{[\mathbf{A}^{-}]_{i-1,q}[\mathbf{A}^{-}]^*_{i,q}\}.  
\end{align}

Obviously, it can be observed from (\ref{eq:JRC_A7}), (\ref{eq:JRC_A14}), and (\ref{eq:JRC_A15}) that the parameters $\widehat{R}_q$, $\widehat{\psi}_q$, $\widehat{r}_q$, and $\widehat{\varphi}_q $ are automatically paired with each other due to the same permutation of $\widehat{\mathbf{T}}$.

\subsection{OAM-based Target Velocity Estimation}
By invoking (\ref{eq:SM_4}), the phase modulation term of received signal for the $q$-th scatterer from the $m$-th transmit element to the $n$-th receive element on the $u$-th OAM state before approximation can be rewritten as 
\begin{align}          \label{eq:JRC_Doppler_1}
    \Psi_q(m,n,l_u) &= - 2\pi\frac{ f_m R_q}{c} - \pi \frac{r_{q}^2- r^2_{\mathrm{max}}(z_{\mathrm{ref}})}{\lambda_m R_{l_u}(z_{\mathrm{ref}})} \notag \\
    &~~~ - l_u \varphi_q  - 2 \pi d \sin(\psi_q)n/{\lambda_m}.
\end{align}
Subsequently, the Doppler frequency shift of the echo signal for the $q$-th scatterer can be derived as
\begin{align}
    f_{D_q}(m,n,l_u) = \frac{1}{2\pi} \frac{\mathrm{d}\Psi_q(m,n,l_u)}{\mathrm{d}t} = f_{LD_q} + f_{RD_q}
\end{align}
where $f_{LD_q}$ is the linear Doppler frequency shift induced by distance variation, including the bi-static distance $R_q$ variation and radial distance $r_q$ variation. $f_{RD_q}$ is the rotational Doppler frequency shift induced by azimuthal angle $\varphi_q$ change. Thus, the phase term corresponding to the linear Doppler frequency shift is
\begin{align}            \label{eq:JRC_Doppler_3}
    \Psi_{LD_q}(m,l_u)= - 2\pi\frac{ f_m R_q}{c} - \pi \frac{r_{q}^2- r^2_{\mathrm{max}}(z_{\mathrm{ref}})}{\lambda_m R_{l_u}(z_{\mathrm{ref}})}
\end{align}
and the phase term corresponding to the rotational Doppler frequency shift is
\begin{align}             \label{eq:JRC_Doppler_4}
    \Psi_{RD_q}(l_u) = - l_u \varphi_q.
\end{align}

Assume that the $q$-th scatterer is moving along the $x$-axis at the linear velocity $\nu_q$. Then, the linear Doppler frequency shift in the OAM-based radar system can be calculated as
\begin{align}             \label{eq:JRC_Doppler_5}
    f_{LD_q}(m,l_u) &= \frac{1}{2\pi} \frac{\mathrm{d}\Psi_{LD_q}(m,l_u)}{\mathrm{d}t}        \notag \\
    &= -\frac{f_m}{c}\frac{\mathrm{d}R_q}{\mathrm{d}t}- \frac{1}{2\lambda_m R_{l_u}(z_{\mathrm{ref}})} \frac{\mathrm{d}r_q^2}{\mathrm{d}t}
\end{align}
where
\begin{align}
    \frac{\mathrm{d}R_q}{\mathrm{d}t} &=   \frac{\mathrm{d}R_q^{(1)}}{\mathrm{d}t}  +  \frac{\mathrm{d}R_q^{(2)}}{\mathrm{d}t}        \notag \\
    &= \frac{\mathrm{d} \sqrt{r_q^2+ z_{\mathrm{ref}}^2} }{\mathrm{d}t} +\nu_q \sin \psi_q            \notag  \\
    &= \frac{r_q \nu_q \cos \varphi_q}{\sqrt{r_q^2+ z_{\mathrm{ref}}^2}} +\nu_q \sin \psi_q  
\end{align}
and
\begin{align}      \label{eq:JRC_Doppler_7}
    \frac{\mathrm{d}r_q^2}{\mathrm{d}t} =2r_q \frac{\mathrm{d}r_q}{\mathrm{d}t} = 2r_q \nu_q \cos \varphi_q.
\end{align}
Based on the fact $\Omega_q r_q =\nu_q \sin \varphi_q$ where $\Omega_q =  \mathrm{d}\varphi_q/\mathrm{d}t $ is the angular velocity of the $q$-th scatterer, the rotational Doppler frequency shift is 
\begin{align}            \label{eq:JRC_Doppler_9}
    f_{RD_q}(l_u) =  \frac{1}{2\pi} \frac{\mathrm{d}\Psi_{RD_q}(l_u)}{\mathrm{d}t} = \frac{-l_u \nu_q \sin\varphi_q}{2\pi r_q}.
\end{align}
Then, the total Doppler frequency shift of the $q$-th scatterer has the form
\begin{align}             \label{eq:JRC_Doppler_10}
    f_{D_q}(m,l_u) &= f_{LD_q}(m,l_u) + f_{RD_q}(l_u)        \notag \\
    &= -\nu_q \left[ \frac{r_q \cos \varphi_q }{\lambda_m \sqrt{r_q^2 + z_{\mathrm{ref}}^2}} + \frac{\sin \psi_q}{\lambda_m} + \frac{r_q \cos \varphi_q}{\lambda_m R_{l_u}(z_{\mathrm{ref}})}  \right.      \notag  \\
    &~~~ \left. +  \frac{l_u \sin\varphi_q}{2\pi r_q} \right].
\end{align}

\subsection{Communication Symbols Recovery}           \label{JRC_Comm}
As stated previously, $(1-\mu)$ percent of OAM carriers are allocated to JRC frames which include the communication symbols. Considering that differential phase shift keying (DPSK) modulation is robust to constant phase shifts, we adopt binary DPSK coding in our system. Let $a_u$ in (\ref{eq:SM_2}) be DPSK symbols, modeled as $a_u= e^{\mathrm{j}\phi_u}$. Then, the received JRC signals are
\begin{align}
    \mathbf{x}_{rc1}(t_s) &= \sum_{q=1}^Q  \widetilde{\rho}_q(t_s) (\mathbf{a}_{c1} \diamond \mathbf{a}_{Trc1}(r_q,\varphi_q)) \otimes  \mathbf{b}(R_q) \otimes  \mathbf{a}_R(\psi_q) \notag \\
         & ~~~ + \mathbf{n}_{rc1}(t_s),     \\
    \mathbf{x}_{rc2}(t_s) &= \sum_{q=1}^Q  \widetilde{\rho}_q(t_s) (\mathbf{a}_{c2} \diamond \mathbf{a}_{Trc2}(r_q,\varphi_q)) \otimes  \mathbf{b}(R_q) \otimes  \mathbf{a}_R(\psi_q) \notag \\
         & ~~~ + \mathbf{n}_{rc2}(t_s),
\end{align}
where $\mathbf{a}_{c1}=[a_{-U},a_{-U+1},\ldots,a_{-\mu U-1}]^T$ is the communication symbol vector modulated by the negative OAM states $\{l_{-U},l_{-U+1},\ldots,l_{-\mu U-1} \}$ and $ \mathbf{a}_{Trc1}(r_q,\varphi_q) \in \mathbb{C}^{(1-\mu) U \times 1}$ has the form
\begin{align}
    \mathbf{a}_{Trc1}(r_q,\varphi_q) = [e^{\mathrm{j}\Phi_q^{l_{-U}}},e^{\mathrm{j}\Phi_q^{l_{-U+1}}},\ldots,e^{\mathrm{j}\Phi_q^{l_{-\mu U-1}}} ]^T.    
\end{align}
Similarly, $\mathbf{a}_{c2}=[a_{\mu U+1},a_{\mu U+2},\ldots,a_{U}]^T$ is the communication symbol vector modulated by the positive OAM states $\{l_{\mu U+1},l_{\mu U+2},\ldots,l_{U} \}$ and $ \mathbf{a}_{Trc2}(r_q,\varphi_q) \in \mathbb{C}^{(1-\mu) U \times 1}$ has the form
\begin{align}
    \mathbf{a}_{Trc2}(r_q,\varphi_q) = [e^{\mathrm{j}\Phi_q^{l_{\mu U+1}}},e^{\mathrm{j}\Phi_q^{l_{\mu U+2}}},\ldots,e^{\mathrm{j}\Phi_q^{l_{U}}} ]^T.    
\end{align}

Define the concatenated vector 
\begin{align}
    \mathbf{x}(t_s)&=  [\mathbf{x}_{rc1}(t_s)^T ~ \mathbf{x}_{r}(t_s)^T~ \mathbf{x}_{rc2}(t_s)^T]^T       \notag \\
                   &= \sum_{q=1}^Q  \widetilde{\rho}_q(t_s) [\mathrm{diag}(\mathbf{a}_{c})\mathbf{a}_{T}(r_q,\varphi_q)]  \otimes  \mathbf{b}(R_q) \otimes  \mathbf{a}_R(\psi_q)  \notag \\
                   &~~~~~ + \mathbf{n}(t_s),
\end{align}
where $\mathbf{a}_c = [\mathbf{a}_{c1}^T ~ \mathbf{1}_{(2\mu U+1)\times 1}^T ~ \mathbf{a}_{c2}^T ]^T \in \mathbb{C}^{(2U+1) \times 1}$ and
\begin{align}
    \mathbf{a}_{T}(r_q, \varphi_q) = \left[ 
                    \begin{matrix}
                        \mathbf{a}_{Trc1}(r_q, \varphi_q)  \\ 
                        \mathbf{a}_{Tr}(r_q, \varphi_q)  \notag \\
                        \mathbf{a}_{Trc2}(r_q, \varphi_q) 
                    \end{matrix}
                   \right]  \in \mathbb{C}^{(2U+1) \times 1}
\end{align}
is the multiplexing transmitting steering vector. $\mathbf{n}(t_s)=[\mathbf{n}_{rc1}^T(t_s)~ \mathbf{n}_r^T(t_s)~\mathbf{n}_{rc2}^T(t_s) ]^T \in \mathbb{C}^{(2U+1)MN \times 1}$. Our proposed OAM-based MDM strategy is depicted in Fig. \ref{MDM_strategy}. Denote $\mathbf{R} \triangleq \mathrm{E}\{ \mathbf{x}(t_s) \mathbf{x}^H(t_s)  \}$ as the multiplexing covariance matrix and $\mathbf{A}_{T}\triangleq [\mathbf{a}_{T}(r_1,\varphi_1),\mathbf{a}_{T}(r_2,\varphi_2),\ldots,\mathbf{a}_{T}(r_Q,\varphi_Q)] \in \mathbb{C}^{(2U+1) \times Q}$. Using the similar steps from (\ref{eq:JRC_A2}) to (\ref{eq:JRC_A2_2}), we can obtain the signal subspace matrix $\mathbf{U}_s$ which has the form
\begin{align}          \label{eq:JRC_B4}
    \mathbf{U}_s &= [\mathrm{diag}(\mathbf{a}_c)\mathbf{A}_{T} \odot \mathbf{B} \odot\mathbf{A}_R]\mathbf{T}  \notag  \\
         &= [\mathrm{diag}(\mathbf{a}_c) \otimes \mathbf{I}_{MN}](\mathbf{A}_{T} \odot \mathbf{B} \odot\mathbf{A}_R) \mathbf{T} .
\end{align}

\begin{figure}[t]
\centerline{\includegraphics[scale=0.40]{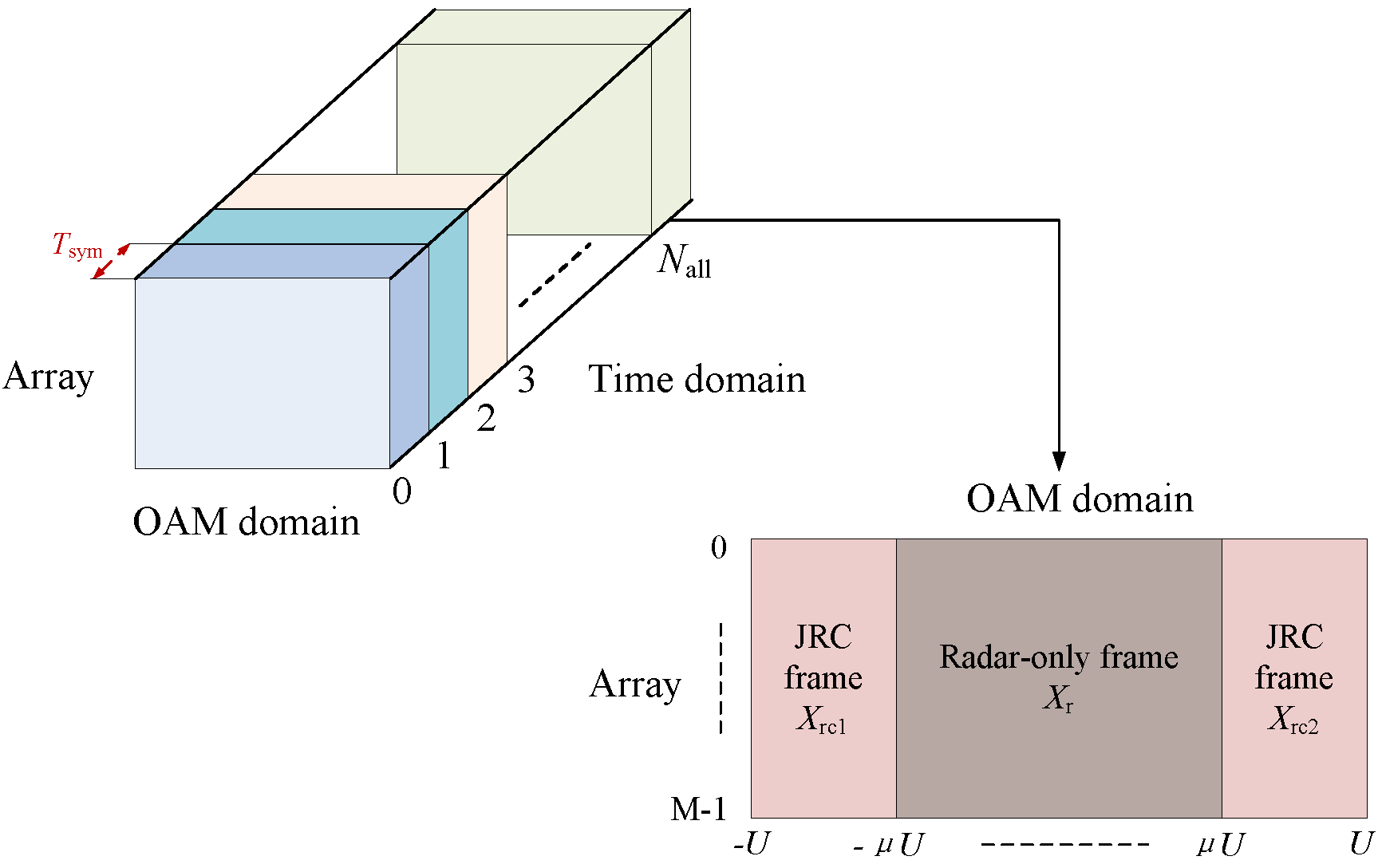}}
\caption{Our proposed OAM-based MDM strategy. The radar-only frame $X_r$ is multiplexed with the joint radar-communications frames $X_{rc1}$ and $X_{rc2}$ in OAM domain. }
\label{MDM_strategy}
\end{figure}


In this model, the maximum ratio combining (MRC) technique, which can fully utilize the transmitting and receiving diversity, is applied to maximize the signal-to-noise ratio (SNR) of the received signals. Considering that the unknown radar parameters and permutation matrix have been estimated in Section \ref{JRC_Radar}, namely $\{  \widehat{r}_q \}_{q=1}^Q$, $\{\widehat{\varphi}_q \}_{q=1}^{Q}$, $\{\widehat{R}_q \}_{q=1}^Q$, $\{\widehat{\psi}_q \}_{q=1}^Q$, and $\widehat{\mathbf{T}}$, we can recover the unknown communication symbols vector $\mathbf{a}_{c1}$ and $\mathbf{a}_{c2}$ by substituting the estimated radar parameters into (\ref{eq:JRC_B4}). Let $\mathbf{U}_s^{(u,i)}$ and $\mathbf{H}_{(u,i)}$ be the $[(u+U)MN+i]$-th row of $\mathbf{U}_s$ and $(\mathbf{A}_{T} \odot \mathbf{B} \odot\mathbf{A}_R) \mathbf{T}$, respectively, with $u \in [-U,-\mu U -1] \cup [\mu U +1, U]$ and $i\in [1,MN]$. Then, the $u$-th communication symbol can be estimated as
\begin{align}        \label{eq:JRC_B5}
    a_u = \frac{1}{MN}\sum_{i=1}^{MN} \frac{\mathbf{U}_s^{(u,i)}\mathbf{H}_{(u,i)}^H}{\mathbf{H}_{(u,i)}\mathbf{H}_{(u,i)}^H}.
\end{align}

\section{Performance Analysis}      \label{sec:PA}
\subsection{Recovery Guarantees of OAM-JRC}
We derive detailed theoretical recovery guarantees in the proposed OAM-JRC system. In particular, the conditions on the number of transmit (receive) array elements and OAM states required to retrieve the unknown parameters of $Q$ scatterers and communication symbols are discussed, which provide the reference for choosing the JRC sharing factor $\mu$. The following Theorem \ref{Theo:RG} provides lower bounds on the number of transmit (receive) array elements and OAM states required by the OAM-JRC system to guarantee perfect recovery of the unknown scatterer parameter set $\{\varphi_q, \psi_q, R_q, r_q \}_{q=1}^Q $ and communication symbol $a_u$.

\begin{theorem}       \label{Theo:RG}
Consider a bistatic FDA system with an $M(N)$-element transmit (receive) ULA and $d$, $\Delta f$ be the interelement spatial spacing and frequency increment, respectively. Each transmit element sends out a total of $2U+1$ OAM states with the inter-state spacing $\delta$. Assume that there are $Q$ scatterers with $R_{\mathrm{max}}$ being the maximum bi-static range and the JRC sharing factor is denoted as $\mu$. If
\begin{description}
\item[C1] $ d\leq \frac{c}{2f_0}$,\vspace{4pt}
\item[C2] $\Delta f \leq \frac{c}{R_{\mathrm{max}}}$, \vspace{4pt}
\item[C3] $ N>Q  $,\vspace{4pt}
\item[C4] $ M>Q $, \vspace{4pt}
\item[C5] $ \frac{1}{U} \leq \mu \leq 1 $,
\end{description}
then the unknown parameter set $\{\varphi_q, \psi_q, R_q, r_q \}_{q=1}^Q$ of $Q$ scatterers and communication symbols $a_u$ can be perfectly recovered from the multiplexing covariance matrix $\mathbf{R}$.
\end{theorem}
\begin{proof}
First of all, we prove that there exists an $Q\times Q$ invertible matrix $\mathbf{T}$ such that (\ref{eq:JRC_A2_2}) holds under \upshape\textbf{C1}-\textbf{C4}. It can be seen from (\ref{eq:JRC_A2}) that $\mathbf{B} $ and $\mathbf{A}_R$ are both Vandermonde matrices which have the size of $M\times Q$ and $N\times Q$, respectively, and they have distinct columns based on the assumption \upshape\textbf{A1} and conditions \upshape\textbf{C1}-\textbf{C2}. Accordingly, if conditions \upshape\textbf{C3}-\textbf{C4} hold, $\mathbf{B} $ and $\mathbf{A}_R$ have full column rank, namely 
\begin{align}            \label{eq:PA_1}
    \mathrm{rank}(\mathbf{B}) = \mathrm{rank}(\mathbf{A}_R) = Q. 
\end{align}
Recall the definition of Kruskal rank as follows: 
\begin{definition}            \label{Def:Krank}
   \cite{harshman1984} The Kruskal rank of a matrix $\mathbf{A}$, denoted as $\mathrm{Krank}(\mathbf{A})$, is the maximum number $\kappa$ such that any $\kappa$ columns of $\mathbf{A}$ are linearly independent. 
\end{definition}

From (\ref{eq:PA_1}) and Definition \ref{Def:Krank}, we have 
\begin{align}
     \mathrm{Krank}(\mathbf{B}) = \mathrm{Krank}(\mathbf{A}_R) = Q.
\end{align}
Then, we evaluate the rank of $\mathbf{B} \odot \mathbf{A}_R$. By invoking the Kruskal rank property of Khatri-Rao product for any two matrices \cite{NDSidiropoulos2000}, it holds that
\begin{align}                   \label{eq:PA_3}
     \mathrm{Krank}(\mathbf{B} \odot \mathbf{A}_R) & \geq \min (\mathrm{Krank}(\mathbf{B}) + \mathrm{Krank}(\mathbf{A}_R)-1, Q)    \notag  \\
         & = \min (2Q-1, Q)     \notag \\
         & = Q.
\end{align}
Considering that $\mathbf{B} \odot \mathbf{A}_R $ is an overdetermined matrix with the size of $MN \times Q$, namely $MN > Q$, it can be obtained that
\begin{align}                   \label{eq:PA_4}
    \mathrm{rank}(\mathbf{B} \odot \mathbf{A}_R) = \mathrm{Krank}(\mathbf{B} \odot \mathbf{A}_R) =Q.
\end{align}
Next, we analyze the rank of $\mathbf{A}_{Tr} \odot \mathbf{B} \odot\mathbf{A}_R$. From the Kruskal-rank property of Khatri-Rao product in (\ref{eq:PA_3}) and the result in (\ref{eq:PA_4}), we have
\begin{align}                   \label{eq:PA_5}
    \mathrm{Krank}(\mathbf{A}_{Tr} \odot \mathbf{B} \odot\mathbf{A}_R) & \geq \min (\mathrm{Krank}(\mathbf{A}_{Tr}) + \mathrm{Krank}(\mathbf{B} \odot\mathbf{A}_R)              \notag \\
    &~~~~~~~~~~  -1, Q)             \notag \\
    & = \min (\mathrm{Krank}(\mathbf{A}_{Tr}) + Q -1, Q).
\end{align}
It is obvious from Definition \ref{Def:Krank} that the Kruskal rank of a matrix $\mathbf{A}$ is greater than or equal to 0. Note that $\mathrm{Krank}(\mathbf{A}) = 0$ if and only if $\mathbf{A}$ contains at least one identically zero column. Hence, for the matrix $\mathbf{A}_{Tr}$ in our case, we have $\mathrm{Krank}(\mathbf{A}_{Tr}) \geq 1$, leading to
\begin{align}
     \mathrm{Krank}(\mathbf{A}_{Tr} \odot \mathbf{B} \odot\mathbf{A}_R) & \geq \min (\mathrm{Krank}(\mathbf{A}_{Tr}) + Q -1, Q)   \notag \\
     & = Q.
\end{align}
Similar to (\ref{eq:PA_4}), for the overdetermined matrix $\mathbf{A}_{Tr} \odot \mathbf{B} \odot\mathbf{A}_R$ with the size of $(2\mu U +1)MN \times Q$, we have
\begin{align}        \label{eq:PA_7}
    \mathrm{rank}(\mathbf{A}_{Tr} \odot \mathbf{B} \odot\mathbf{A}_R) = \mathrm{Krank}(\mathbf{A}_{Tr} \odot \mathbf{B} \odot\mathbf{A}_R) =Q.
\end{align}
Applying eigen-decomposition in (\ref{eq:JRC_A3_1}), we can get
\begin{align}
    (\mathbf{A}_{Tr} \odot \mathbf{B} \odot\mathbf{A}_R)\mathbf{R}_{\rho} (\mathbf{A}_{Tr} \odot \mathbf{B} \odot\mathbf{A}_R)^H \mathbf{U}_{rn} =0.
\end{align}
Based on (\ref{eq:PA_7}) and considering that $\mathbf{R}_{\rho}$ is a diagonal matrix which has full rank with $ \mathrm{rank}(\mathbf{R}_{\rho})=Q$, it holds that 
\begin{align}
    (\mathbf{A}_{Tr} \odot \mathbf{B} \odot\mathbf{A}_R)^H \mathbf{U}_{rn}=0
\end{align}
which implies that $\mathbf{U}_{rn}$ is the null space of $\mathbf{A}_{Tr} \odot \mathbf{B} \odot\mathbf{A}_R$ and $\mathbf{U}_{rs}$. Consequently, there exists an $Q\times Q$ invertible matrix $\mathbf{T}$ such that (\ref{eq:JRC_A2_2}) holds. 

Secondly, we prove that $\mathbf{U}_{B_1}$ in (\ref{eq:JRC_A6_1}) and $\mathbf{U}_{A_{R1}}$ in (\ref{eq:JRC_A6_2}) are left-invertible. By using the similar steps from (\ref{eq:PA_5}) to (\ref{eq:PA_7}) under \upshape\textbf{C1}-\textbf{C4} and \upshape\textbf{A1}, we have 
\begin{align}
    \mathrm{rank}(\mathbf{A}_{Tr} \odot \mathbf{B} \odot\mathbf{A}_{R1}) = \mathrm{Krank}(\mathbf{A}_{Tr} \odot \mathbf{B} \odot\mathbf{A}_{R1}) =Q
\end{align}
and
\begin{align}
    \mathrm{rank}(\mathbf{A}_{Tr} \odot \mathbf{B}_1 \odot\mathbf{A}_R) = \mathrm{Krank}(\mathbf{A}_{Tr} \odot \mathbf{B}_1 \odot\mathbf{A}_R) =Q.
\end{align}
Considering that $ \mathbf{T} \in \mathbb{C}^{Q\times Q}$ is non-singular, we can get
\begin{align}
    \mathrm{rank}(\mathbf{U}_{B_1}) = \mathrm{rank}[(\mathbf{A}_{Tr} \odot \mathbf{B}_1 \odot\mathbf{A}_R) \mathbf{T}] = Q 
\end{align}
and
\begin{align}
    \mathrm{rank}(\mathbf{U}_{A_{R1}}) = \mathrm{rank}[(\mathbf{A}_{Tr} \odot \mathbf{B} \odot\mathbf{A}_{R1}) \mathbf{T}] = Q 
\end{align}
which imply that $\mathbf{U}_{B_1} $ and  $\mathbf{U}_{A_{R1}}$ are both full column column rank, namely left-invertible. Therefore,  the bi-static ranges $\{R_q\}_{q=1}^Q$ and DoAs $\{\psi_q\}_{q=1}^Q$ can be recovered uniquely. 

As soon as $\{R_q\}_{q=1}^Q$ and $\{\psi_q\}_{q=1}^Q$ are estimated, $\mathbf{A}^{+}$ and $\mathbf{A}^{-}$ can be calculated based on (\ref{eq:JRC_A10}) and (\ref{eq:JRC_A11}). Then, the azimuthal angles $\{\varphi_q \}_{q=1}^Q$ and radial distances $\{r_q\}_{q=1}^Q$ can be recovered uniquely from (\ref{eq:JRC_A14}) and (\ref{eq:JRC_A15}) as long as $\mu U +1 \geq 2$, namely $\mu \geq \frac{1}{U}$. Considering that $\mu$ is defined as the JRC sharing factor with $\mu \leq 1$, the sufficient condition for unique recovery of $\{\varphi_q \}_{q=1}^Q$ and $\{r_q\}_{q=1}^Q$ can be rewritten as \textbf{C5}: $\frac{1}{U} \leq \mu \leq 1 $. 

Obviously, all the recovered target position parameters $\{\varphi_q, \psi_q, R_q, r_q \}_{q=1}^Q $ are auto-paired due to the same permutation of $\mathbf{T}$. By substituting the estimated $\{\widehat{\varphi}_q, \widehat{\psi}_q, \widehat{R}_q, \widehat{r}_q \}$  into (\ref{eq:JRC_B4}) and (\ref{eq:JRC_B5}), the communication symbol vectors $\mathbf{a}_{c1}$ and $\mathbf{a}_{c2}$ can be recovered uniquely. 

\end{proof} 

\subsection{CRLB of Radar Target Parameters}
In this subsection, we explore the statistical bounds of radar target parameters, namely the CRLB of position parameters $\{\varphi_q, \psi_q, R_q, r_q \}_{q=1}^Q$ and velocity parameters $\{\nu_q,\Omega_q  \}_{q=1}^Q$. 

Firstly, we derive the CRLB of unknown position parameters $\bm{\gamma} = [\bm{\varphi}^T~ \bm{r}^T ~ \bm{R}^T~ \bm{\psi}^T]^T$, where $\bm{\varphi}= [\varphi_1,\ldots,\varphi_Q]^T$, $\bm{r} = [r_1,\ldots,r_Q]^T$, $\bm{R}=[R_1,\ldots,R_Q]^T$ and $ \bm{\psi} = [\psi_1,\ldots,\psi_Q]^T$. Here, the radar-only signal model in (\ref{eq:JRC_A1}) is considered. Since $ \mathbf{R}_r$ in (\ref{eq:JRC_A3_1}) is unavailable in practice, it can be substituted by its asymptotic unbiased estimate from $L_r$ available snapshots, namely
\begin{align}                \label{eq:CRLB_1}
    \widehat{\mathbf{R}}_r = \frac{1}{L_r} \sum_{t_s=0}^{L_r-1}  \mathbf{x}_r(t_s) \mathbf{x}^H_r(t_s).
\end{align}
The received signal vectors $\{  \mathbf{x}_r(t_s) \}_{t_s=1}^{L_r}$ are assumed to be i.i.d. Gaussian random variables, distributed as
\begin{align}                \label{eq:CRLB_1_1}
    \mathbf{x}_r(t_s) \sim \mathcal{CN}(\mathbf{0},\mathbf{R}_r).
\end{align}
Considering that the covariance matrix $\mathbf{R}_r$ is parameterized by $\bm{\gamma}$, the Fisher information matrix (FIM) $\mathbf{F}(\bm{\gamma})$ follows directly from the Slepian-Bangs formula \cite{PJSchreier2010} as 
\begin{align}         \label{eq:CRLB_2}
    \frac{1}{L_r} [\mathbf{F}(\bm{\gamma})] &= \mathrm{vec}^H \left( \frac{\partial \mathbf{R}_r}{\partial \bm{\gamma}^T} \right) \mathbf{W} \mathrm{vec} \left( \frac{\partial \mathbf{R}_r}{\partial \bm{\gamma}^T} \right)       \notag \\
    & = \left( \frac{\partial \mathrm{vec}(\mathbf{R}_r)}{\partial \bm{\gamma}^T} \right)^H  \mathbf{W}   \left( \frac{\partial \mathrm{vec}(\mathbf{R}_r)}{\partial \bm{\gamma}^T} \right)  \notag \\
    &= \left[\frac{\partial \mathbf{r}_r}{\partial \bm{\varphi}^T} ~ \frac{\partial \mathbf{r}_r}{\partial \bm{r}^T} ~ \frac{\partial \mathbf{r}_r}{\partial \bm{R}^T}~ \frac{\partial \mathbf{r}_r}{\partial \bm{\psi}^T} \right]^H  \mathbf{W}       \notag \\
    &~~~ \times \left[\frac{\partial \mathbf{r}_r}{\partial \bm{\varphi}^T} ~ \frac{\partial \mathbf{r}_r}{\partial \bm{r}^T} ~ \frac{\partial \mathbf{r}_r}{\partial \bm{R}^T}~ \frac{\partial \mathbf{r}_r}{\partial \bm{\psi}^T} \right]      \notag \\
    & = \left[\mathbf{D}_{\bm{\varphi}} ~ \mathbf{D}_{\bm{r}} ~ \mathbf{D}_{\bm{R}} ~ \mathbf{D}_{\bm{\psi}} \right]^H  \mathbf{W} \left[\mathbf{D}_{\bm{\varphi}} ~ \mathbf{D}_{\bm{r}} ~ \mathbf{D}_{\bm{R}} ~ \mathbf{D}_{\bm{\psi}} \right]
\end{align}
where $\mathbf{W} = \mathbf{R}_r^{-T} \otimes \mathbf{R}_r^{-1}$, $\mathbf{r}_r = \mathrm{vec}(\mathbf{R}_r)=  [(\mathbf{A}_{Tr} \odot \mathbf{B} \odot\mathbf{A}_R)^* \odot (\mathbf{A}_{Tr} \odot \mathbf{B} \odot\mathbf{A}_R)] \mathbf{r}_{\rho}$ with $\mathbf{r}_{\rho} = [\sigma_1^2, \sigma_2^2,\ldots,\sigma_Q^2]^T$, $\mathbf{D}_{\bm{\varphi}} = \partial \mathbf{r}_r/ \partial \bm{\varphi}^T$, $\mathbf{D}_{\bm{r}} = \partial \mathbf{r}_r / \partial \bm{r}^T$, $\mathbf{D}_{\bm{R}} = \partial \mathbf{r}_r /\partial \bm{R}^T$ and $\mathbf{D}_{\bm{\psi}} = \partial \mathbf{r}_r / \partial \bm{\psi}^T$. For ease of expression, denote $\bm{l}_{\mathcal{M}} \triangleq [l_{-U},\ldots,l_{-1},l_0,l_1,\ldots,l_U]^T$ as the OAM state number vector and $\widetilde{\mathbf{R}}_{\mathcal{M}} \triangleq [1/R_{l_{-U}}(z_{\mathrm{ref}}),1/R_{l_{-U+1}}(z_{\mathrm{ref}}),\ldots,1/R_{l_{U}}(z_{\mathrm{ref}})]^T$. Then, the specific forms of $\mathbf{D}_{\bm{\varphi}}$, $\mathbf{D}_{\bm{r}}$, $\mathbf{D}_{\bm{R}}$ and $\mathbf{D}_{\bm{\psi}}$ can be expressed as (\ref{eq:CRLB_3})$\sim$(\ref{eq:CRLB_6}), respectively. After obtaining $\mathbf{F}(\bm{\gamma})$ from (\ref{eq:CRLB_2}), the CRLB of the position parameters $\bm{\gamma}$ to be estimated can be expressed as 
\begin{align}
    \mathrm{E}\left\{(\gamma_n - \widehat{\gamma}_n)^2\right\} \geq \mathrm{CRLB}(\gamma_n) = [\mathbf{F}^{-1}(\bm{\gamma})]_{n,n}. 
\end{align}

\begin{subequations}
\begin{figure*}[!t]            
\par\noindent\small
\begin{align}      \label{eq:CRLB_3}
    \mathbf{D}_{\bm{\varphi}} & = \left[\left( \frac{\partial \mathbf{A}_{Tr}^*}{\partial \bm{\varphi}^T} \odot  \mathbf{B}^* \odot\mathbf{A}_R^* \right) \odot \left(\mathbf{A}_{Tr} \odot \mathbf{B} \odot\mathbf{A}_R \right) +  \left(\mathbf{A}_{Tr}^* \odot \mathbf{B}^* \odot\mathbf{A}_R^* \right) \odot \left( \frac{\partial \mathbf{A}_{Tr}}{\partial \bm{\varphi}^T} \odot  \mathbf{B} \odot\mathbf{A}_R \right) \right]  \mathbf{R}_{\rho}  \notag \\ 
    &  = \left\{\left[\mathrm{diag}(\mathrm{j}\bm{l}_{\mathcal{M}})\mathbf{A}_{Tr}^*  \odot  \mathbf{B}^* \odot\mathbf{A}_R^* \right] \odot \left(\mathbf{A}_{Tr} \odot \mathbf{B} \odot\mathbf{A}_R \right) +  \left(\mathbf{A}_{Tr}^* \odot \mathbf{B}^* \odot\mathbf{A}_R^* \right) \odot \left[\mathrm{diag}(-\mathrm{j}\bm{l}_{\mathcal{M}}) \mathbf{A}_{Tr} \odot  \mathbf{B} \odot\mathbf{A}_R \right] \right\}  \mathbf{R}_{\rho} ,
\end{align}
\begin{align}      \label{eq:CRLB_4}
    \mathbf{D}_{\bm{r}} &= \left[\left( \frac{\partial \mathbf{A}_{Tr}^*}{\partial \bm{r}^T} \odot  \mathbf{B}^* \odot\mathbf{A}_R^* \right) \odot \left(\mathbf{A}_{Tr} \odot \mathbf{B} \odot\mathbf{A}_R \right) +  \left(\mathbf{A}_{Tr}^* \odot \mathbf{B}^* \odot\mathbf{A}_R^* \right) \odot \left( \frac{\partial \mathbf{A}_{Tr}}{\partial \bm{r}^T} \odot  \mathbf{B} \odot\mathbf{A}_R \right) \right]  \mathbf{R}_{\rho}  \notag \\  
    &= \left\{\left[ \mathrm{j}\frac{2\pi}{\lambda_0} (\widetilde{\mathbf{R}}_{\mathcal{M}}\bm{r}^T \diamond \mathbf{A}_{Tr}^*) \odot  \mathbf{B}^* \odot\mathbf{A}_R^* \right] \odot \left(\mathbf{A}_{Tr} \odot \mathbf{B} \odot\mathbf{A}_R \right) +  \left(\mathbf{A}_{Tr}^* \odot \mathbf{B}^* \odot\mathbf{A}_R^* \right) \odot \left[ -\mathrm{j}\frac{2\pi}{\lambda_0} (\widetilde{\mathbf{R}}_{\mathcal{M}}\bm{r}^T \diamond \mathbf{A}_{Tr})  \odot  \mathbf{B} \odot\mathbf{A}_R \right] \right\}  \mathbf{R}_{\rho} ,
\end{align}
\begin{align}      \label{eq:CRLB_5}
     \mathbf{D}_{\bm{R}} &= \left[\left(\mathbf{A}_{Tr}^* \odot \frac{ \partial \mathbf{B}^*}{\partial \bm{R}^T} \odot\mathbf{A}_R^* \right) \odot \left(\mathbf{A}_{Tr} \odot \mathbf{B} \odot\mathbf{A}_R \right) +  \left(\mathbf{A}_{Tr}^* \odot \mathbf{B}^* \odot\mathbf{A}_R^* \right) \odot \left(\mathbf{A}_{Tr} \odot \frac{ \partial \mathbf{B}}{\partial \bm{R}^T} \odot\mathbf{A}_R \right) \right]  \mathbf{R}_{\rho}  \notag \\
    &= \left\{\left[\mathbf{A}_{Tr}^* \odot \left( \mathbf{B}^* \diamond \frac{ \ln \mathbf{B}^*}{ \mathrm{diag}(\bm{R})} \right)\odot\mathbf{A}_R^* \right] \odot \left(\mathbf{A}_{Tr} \odot \mathbf{B} \odot\mathbf{A}_R \right) + \left(\mathbf{A}_{Tr}^* \odot \mathbf{B}^* \odot\mathbf{A}_R^* \right) \odot \left[\mathbf{A}_{Tr} \odot \left( \mathbf{B} \diamond \frac{ \ln \mathbf{B}}{ \mathrm{diag}(\bm{R})} \right)\odot\mathbf{A}_R \right] \right\}  \mathbf{R}_{\rho},
\end{align}
\begin{align}       \label{eq:CRLB_6}
     \mathbf{D}_{\bm{\psi}} &= \left[\left(\mathbf{A}_{Tr}^* \odot \mathbf{B}^* \odot \frac{\partial \mathbf{A}_R^*}{\partial \bm{\psi}^T} \right) \odot \left(\mathbf{A}_{Tr} \odot \mathbf{B} \odot\mathbf{A}_R \right) +  \left(\mathbf{A}_{Tr}^* \odot \mathbf{B}^* \odot\mathbf{A}_R^* \right) \odot \left(\mathbf{A}_{Tr} \odot \mathbf{B} \odot \frac{\partial \mathbf{A}_R}{\partial \bm{\psi}^T} \right) \right]  \mathbf{R}_{\rho}  \notag \\
     &= \left\{ \left[ \mathbf{A}_{Tr}^* \odot \mathbf{B}^* \odot \left( \mathbf{A}_R^* \diamond \frac{\mathrm{ln}\mathbf{A}_R^*}{\mathrm{diag}(\tan \bm{\psi} )}  \right) \right] \odot \left(\mathbf{A}_{Tr} \odot \mathbf{B} \odot\mathbf{A}_R \right) +  \left(\mathbf{A}_{Tr}^* \odot \mathbf{B}^* \odot\mathbf{A}_R^* \right) \odot \left[ \mathbf{A}_{Tr} \odot \mathbf{B} \odot \left( \mathbf{A}_R \diamond \frac{\mathrm{ln}\mathbf{A}_R}{\mathrm{diag}(\tan \bm{\psi} )}  \right) \right] \right\}  \mathbf{R}_{\rho} 
\end{align} \normalsize
\hrulefill
\end{figure*}
\end{subequations}

Secondly, we explore the CRLB of velocity parameters $\{\nu_q,\Omega_q  \}_{q=1}^Q$ in the proposed OAM-based JRC system. Denote $\bm{\nu}=[\nu_1,\ldots,\nu_Q]^T$. From (\ref{eq:JRC_Doppler_5}) and (\ref{eq:JRC_Doppler_10}), the Doppler frequencies depend on the transmit element index $m$ and the OAM state index $l_u$. Considering that $\Omega_q$ and $\nu_q$ satisfy the one-to-one mapping relationship, we mainly focus on the CRLB derivation of $\{\nu_q\}_{q=1}^Q$ here. During the CRLB derivation of position parameters, the velocity terms are not included in (\ref{eq:CRLB_2}) due to the assumption that all position parameters $\{\varphi_q, \psi_q, R_q, r_q \}_{q=1}^Q $ are time-independent while constructing the covariance matrix $\mathbf{R}_r$ in (\ref{eq:CRLB_1}). However, in practice, microchanges exist in azimuthal angles $\{\varphi_q\}_{q=1}^Q$ and ranges $\{R_q, r_q \}_{q=1}^Q$, caused by target velocities. Thus, we utilize the original received baseband signal model in (\ref{eq:SM_4}) for the following analysis. By substituting (\ref{eq:JRC_Doppler_1}), (\ref{eq:JRC_Doppler_3}) and (\ref{eq:JRC_Doppler_4}) into (\ref{eq:SM_4}), we can rewrite $x_{m,n}^{l_u}(t_s) $ as 
\begin{align}
    x_{m,n}^{l_u}(t_s) &= \sum_{q=1}^Q  \acute{\rho}_q(t_s) a_u  e^{\mathrm{j}[\Psi_{LD_q}(m,l_u) + \Psi_{RD_q}(l_u) -2 \pi d \sin(\psi_q)n/{\lambda_m} ] } \notag \\
    & ~~~ + n_{m,n}^{l_u}(t_s) .           
\end{align}
By using the similar steps from (\ref{eq:CRLB_1_1}) to (\ref{eq:CRLB_2}), the FIM with respect to the Doppler parameters $\bm{\nu}=[\nu_1,\ldots,\nu_Q]^T$ can be expressed as
\begin{align}       \label{eq:CRLB_8}
    \mathbf{F}_{\mathrm{D}}(\bm{\nu}) &= \left[\frac{\partial \mathbf{r}_r}{\partial \bm{\nu}^T} \right]^H \mathbf{W} \left[\frac{\partial \mathbf{r}_r}{\partial \bm{\nu}^T} \right]  \notag \\
    &= \left[\frac{\partial \mathbf{r}_r}{\partial \bm{r}^T} \frac{\partial \bm{r}^T }{\partial\bm{\nu}^T}  \right]^H \mathbf{W} \left[\frac{\partial \mathbf{r}_r}{\partial \bm{r}^T} \frac{\partial \bm{r}^T }{\partial\bm{\nu}^T} \right]   
\end{align}
where $\partial \mathbf{r}_r / \partial \bm{r}^T = \mathbf{D}_{\bm{r}} $ is given in (\ref{eq:CRLB_4}) and $\partial \bm{r}^T / \partial\bm{\nu}^T = t \cos \bm{\varphi}^T$ can be derived from (\ref{eq:JRC_Doppler_7}). Once the FIM of Doppler parameters is estimated from (\ref{eq:CRLB_8}), the corresponding CRLB can be calculated by applying the inverse of $\mathbf{F}_{\mathrm{D}}(\bm{\nu})$.


\subsection{BER Analysis}
Denote the total communication symbol vector as $\widetilde{\mathbf{a}}_c = [\mathbf{a}_{c1}^T ~ \mathbf{a}_{c2}^T]^T \in \mathbb{C}^{2(1-\mu)U \times 1}$. Obviously, the number of OAM states used for the communication function is configured as $2(1-\mu)U$, leading to $\log_2 2(1-\mu)U$ bits information which can be transmitted by selecting one state from the $2(1-\mu)U$ OAM states. Each OAM signal is modulated by a classical binary DPSK modulation, i.e., $a_u= e^{\mathrm{j}\phi_u}$ with $\phi_u \in \{0,\pi \}$ which can bring $\log_2 2 =1$ bit information. Totally, $\log_2 2(1-\mu)U +1$ bits information can be transmitted by an OAM signal in our proposed JRC system. 
To demodulate the communication signals correctly, the OAM states of the signals need to be estimated correctly, then the corresponding radiated symbol should be estimated correctly. Assume that correct demodulation events about the OAM state and radiated symbol are denoted as $\Delta_{\mathrm{OAM}}$ and $\Delta_{\mathrm{DPSK}}$, respectively. If there exist any errors in OAM or radiated symbol demodulation events, the total communication demodulation process is regarded as a failure. Consequently, the probability of the total correct demodulation event $\Delta$ is expressed by
\begin{align}                   \label{eq:BER_1}
    P(\Delta)= P(\Delta_{\mathrm{DPSK}} \vert \Delta_{\mathrm{OAM}}) P(\Delta_{\mathrm{OAM}})
\end{align}
where $P(\Delta_{\mathrm{OAM}})$ is the correct estimation probability of OAM states. $P(\Delta_{\mathrm{DPSK}} \vert \Delta_{\mathrm{OAM}})$ is the conditional correct estimation probability of the transmitted signal waveforms while the OAM states have already been demodulated correctly.

For the calculation of $P(\Delta_{\mathrm{OAM}})$, the pairwise error probability (PEP) $P(l_u \rightarrow l_{\tilde{u}} )$, which represents the error probability where the transmitting OAM state is $l_u$ while the receiving OAM state is detected as $l_{\tilde{u}}$ by mistake, is utilized here. Specifically, $P(l_u \rightarrow l_{\tilde{u}} )$ has the form as \cite{EBasar2018}
\begin{align}
    P(l_u \rightarrow l_{\tilde{u}} ) = \mathbb{Q} \left( \sqrt{\frac{(1-\mu)U \sigma^2_{s,l_u}(1-C_{l_u,l_{\tilde{u}}})}{\sigma_n^2}}  \right)
\end{align}
where $\mathbb{Q}(x)$ is the complementary cumulative distribution function defined as \cite{ABehnad2020}
\begin{align}
    \mathbb{Q}(x) = \frac{1}{\sqrt{2\pi}} \int_x^{\infty} \mathrm{exp}\left(-\frac{u^2}{2}\right) \mathrm{d}u
\end{align}
and $\sigma^2_{s,l_u}$ is the attenuated power at the receiving end which propagates from the transmitting signal with OAM state $l_u$. $C_{l_u,l_{\tilde{u}}} = \cos[(l_u-l_{\tilde{u}})\varphi]$ denotes the correlation coefficient between the two OAM states. Then, $P(\Delta_{\mathrm{OAM}})$ can be obtained as \cite{XXiong2022}
\begin{align}       \label{eq:BER_5}
    P(\Delta_{\mathrm{OAM}}) = 1- \sum_{l_u \neq l_{\tilde{u}}} P(l_u \rightarrow l_{\tilde{u}} ). 
\end{align}

For the calculation of $P(\Delta_{\mathrm{DPSK}} \vert \Delta_{\mathrm{OAM}})$, we assume that the wireless channel is non-selective and has predominantly Rayleigh fading characteristics. Define the signal-to-noise ratio per bit as
\begin{align}
    \varUpsilon_b  = \frac{E_b}{N_0} = \frac{E_s}{N_0 \log_2 (M_I)}
\end{align}
where $E_b$ and $E_s$ are the energy per bit and symbol energy, respectively. $N_0$ is the noise power spectral density. $M_I$ is the number of waveforms which is called modulation index, i.e., for binary DPSK, $M_I=2$. Then, the probability density function of $\varUpsilon_b$ can be obtained to be
\begin{align}
    P(\varUpsilon_b) = \frac{1}{\overline{\varUpsilon}_b} \mathrm{exp}\left( -\frac{\varUpsilon_b}{\overline{\varUpsilon}_b} \right)
\end{align}
where $\overline{\varUpsilon}_b= \mathrm{E}[\varUpsilon_b]$. By invoking that the BER in the presence of additive white Gaussian noise‌ (AWGN) \cite{MTCore2002} is   
\begin{align}
    P_{\mathrm{DPSK}}(\varUpsilon_b) = \frac{1}{2}\mathrm{exp}(-\varUpsilon_b), 
\end{align}
the average BER in a Rayleigh fading channel can be derived as
\begin{align}
    P_b = \int_0^{\infty} P_{\mathrm{DPSK}}(\varUpsilon_b) P(\varUpsilon_b) \mathrm{d} \varUpsilon_b  = \frac{0.5}{1+ \overline{\varUpsilon}_b},
\end{align}
leading to the following expression of $P(\Delta_{\mathrm{DPSK}} \vert \Delta_{\mathrm{OAM}})$ as
\begin{align}               \label{eq:BER_9}
    P(\Delta_{\mathrm{DPSK}} \vert \Delta_{\mathrm{OAM}})= 1- P_b.
\end{align}
By combining (\ref{eq:BER_5}), (\ref{eq:BER_9}) with (\ref{eq:BER_1}), the error probability of our proposed OAM-JRC system can be derived as (\ref{eq:BER_10}).
\begin{figure*}[htp]            
\par\noindent
\begin{align}                \label{eq:BER_10}
    P_e &= 1-  P(\Delta) = 1- P(\Delta_{\mathrm{DPSK}} \vert \Delta_{\mathrm{OAM}}) P(\Delta_{\mathrm{OAM}})   \notag \\
      &= 1- \left\{\left[  1- \sum_{l_u \neq l_{\tilde{u}}} \mathbb{Q} \left( \sqrt{\frac{(1-\mu)U \sigma^2_{s,l_u}(1-C_{l_u,l_{\tilde{u}}})}{\sigma_n^2}}  \right) \right] \left( 
      1-  \frac{0.5}{1+ \overline{\varUpsilon}_b} \right)        \right\}
\end{align}
\hrulefill
\end{figure*}

\section{Simulation Experiments}        \label{sec:SE}
In this section, various simulation experiments are carried out to validate the theoretical derivation and demonstrate the feasibility of the proposed methods. Unless otherwise noted, we set the number of transmit and receive elements as $M=N=6$. The carrier frequency is $f_0=79\mathrm{GHz}$ with the unit frequency increment fixed to $\Delta f = 30\mathrm{kHz}$. We choose $2U+1=21$ OAM states with $\{-10,-9,\ldots,10 \}$ and the beam waist radius corresponding to OAM state $l_u=1$ is set as $w_1=2\lambda_0$ where $\lambda_0=c/f_0$. The DPSK symbol duration is set as $T_{\mathrm{sym}}=5 \mathrm{\mu s}$. Considering that the widths of standard lanes and cars are 3.5m and 1.6m, respectively, the propagation distance can be obtained as  $z_{\mathrm{ref}}=3.5\mathrm{m} - 1.6\mathrm{m}= 1.9\mathrm{m}$. Assume that there are $Q=3$ point scatterers with radar position parameters $\{\varphi, r, R, \psi\}$ as $\{5^{\circ}, 0.21\mathrm{m}, 10\mathrm{m}, 45^{\circ}\}$, $\{30^{\circ}, 0.215\mathrm{m}, 15 \mathrm{m}, 60^{\circ}\}$, and $\{60^{\circ}, 0.22\mathrm{m}, 7\mathrm{m}, 30^{\circ}\}$, respectively and the corresponding velocity parameters as $\bm{\nu}=[10\mathrm{m/s} ~ 20\mathrm{m/s}~ 5\mathrm{m/s} ]$.

\begin{figure}[H]
\centerline{\includegraphics[scale=0.4]{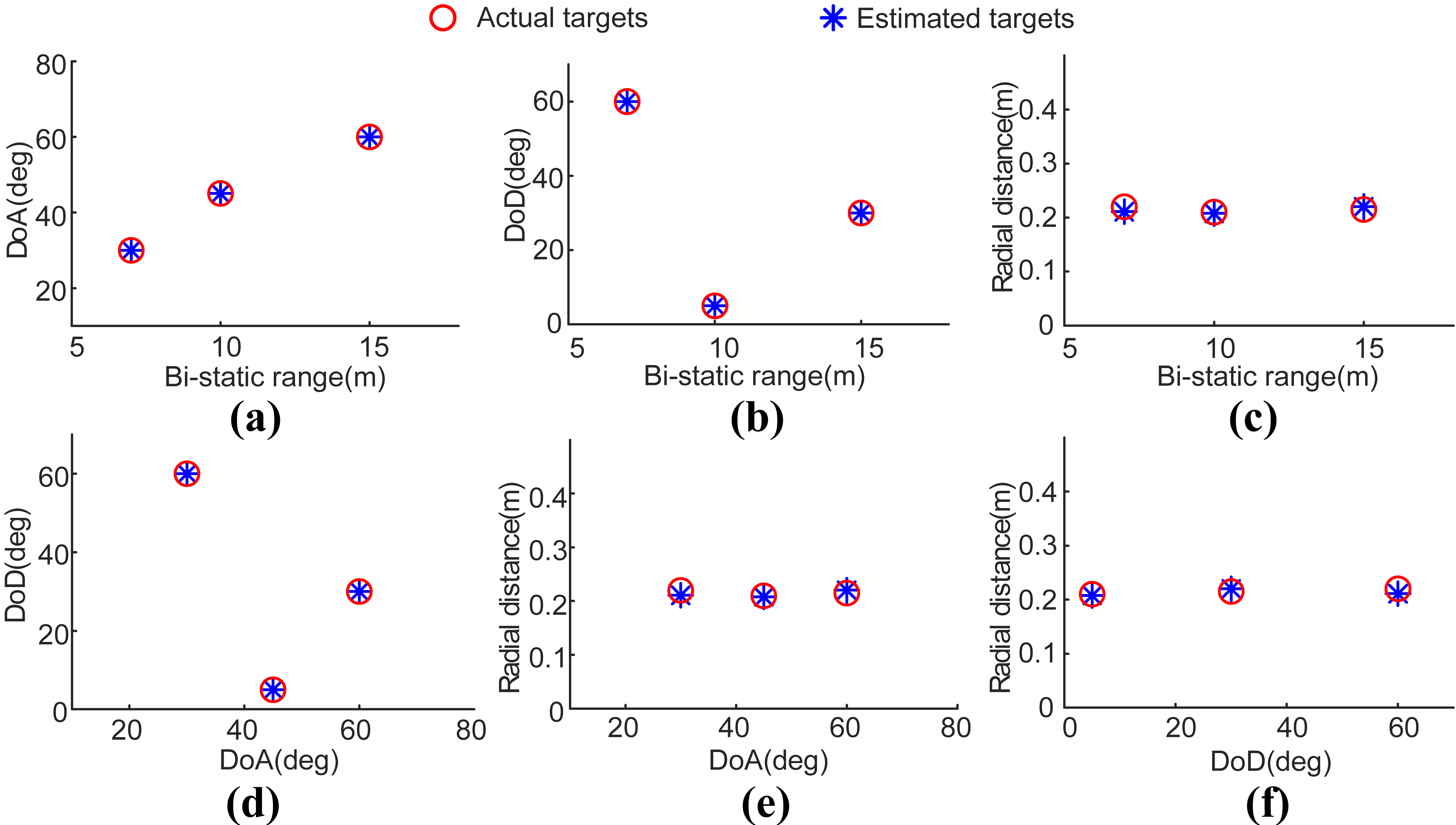}}
\caption{Target detection results of the proposed methods with $Q=3$, $\mu=0.5$ and SNR=10 dB in (a) DoA-bistatic range, (b) DoD-bistatic range, (c) Radial distance-bistatic range, (d) DoD-DoA, (e) Radial distance-DoA, and (f) Radial distance-DoD planes. The red circles indicate actual targets and the blue crosses indicate the estimated targets.}
\label{Radar_3targets}
\end{figure}

\textbf{Radar performance:} We first examine the ability of our proposed OAM-based JRC method to recover the target position parameters in the angle-range domain. By setting the multiplexing parameter $\mu$ as 0.5 and SNR as 10 dB, it can be seen from Fig. \ref{Radar_3targets} that all the target parameters can be recovered perfectly. Furthermore, we benchmark the performance of various multiplexing parameters using the root-mean-squared error (RMSE), defined as
\begin{align}
    \mathrm{RMSE}=\frac{1}{J}\sum_{j=1}^J\sqrt{\frac{1}{Q}\sum_{q=1}^Q\left(\hat{\eta}_{q,j}-\eta_q\right)^2}
\end{align}
where $J$ (set to 500) is the number of Monte Carlo experiments, $\hat{\eta}_{q,j}$ is the estimated parameter (DoD, DoA, bi-static range, or radial distance) of the $q$-th target scatterer in the $j$-th Monte Carlo experiment and $\eta_q$ is the corresponding true value. Fig. \ref{RMSEvsSNR}a, b, c, and d display the RMSE of bi-static range, DoA, DoD, and radial distance, respectivley, versus the received SNR that was varied from -10 to 20 dB in steps of 3 dB for various $\mu$ values. For comparison, we also plot the corresponding root of CRLB (RCRLB) values. It can be seen from Fig. \ref{RMSEvsSNR} that the RMSEs of all radar parameters decrease and approach the RCRLBs gradually with the SNR increasing. Moreover, we find that larger $\mu$ leads to better parameter estimation performance. 

\begin{figure*}
\centerline{\includegraphics[width=1\textwidth]{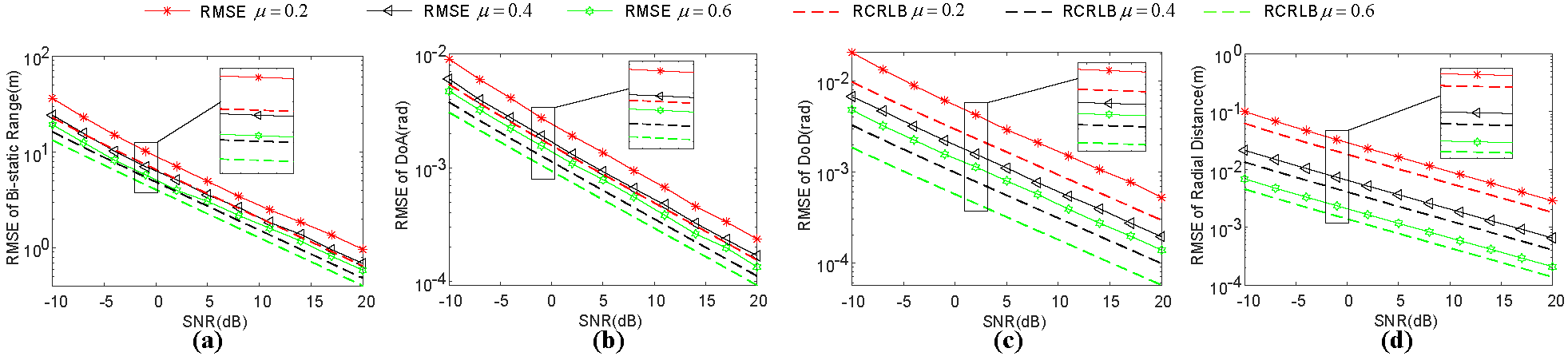}}
\caption{RMSE of (a) bi-static range, (b) DoA, (c) DoD, and (d) radial distance versus SNR for various $\mu$ values together with the corresponding RCRLB values.}
\label{RMSEvsSNR}
\end{figure*}

Similarly, we also validate the effectiveness of the proposed OAM-based target velocity estimation method. By setting a short time window, we can obtain the Doppler frequency spectrum of OAM-based radar in Fig. \ref{Doppler_curve}. Once the Doppler frequencies are estimated, the target velocity parameters can be recovered based on (\ref{eq:JRC_Doppler_10}). Fig. \ref{RMSE_nu_vs_SNR} displays the RMSE of target velocity parameters versus SNR for various $\mu$ and the corresponding RCRLB values where the same conclusion can be reached as that in Fig. \ref{RMSEvsSNR}.

\begin{figure}[t]
\centerline{\includegraphics[scale=0.43]{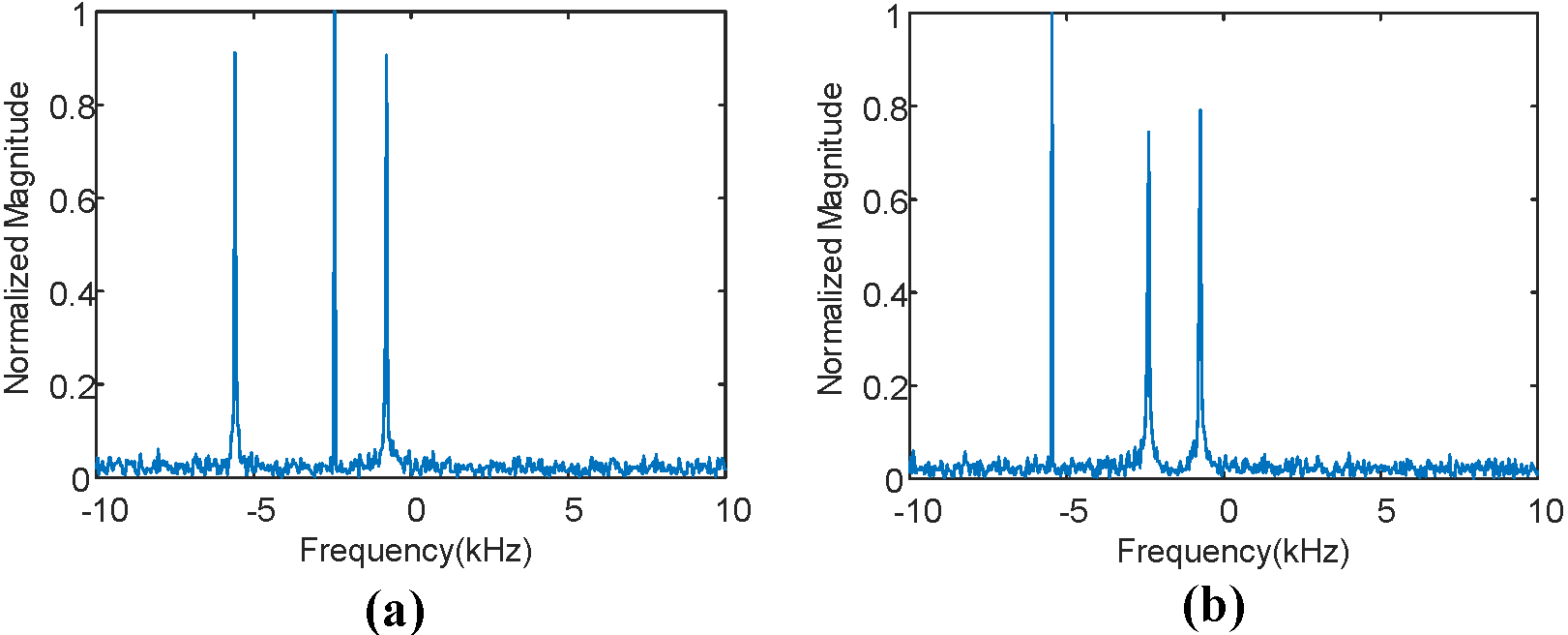}}
\caption{Doppler frequency spectrum with SNR=10 dB at (a) $l_u=0$ and (b) $l_u=-5$.}
\label{Doppler_curve}
\end{figure}
\begin{figure}[t]
\centerline{\includegraphics[scale=0.51]{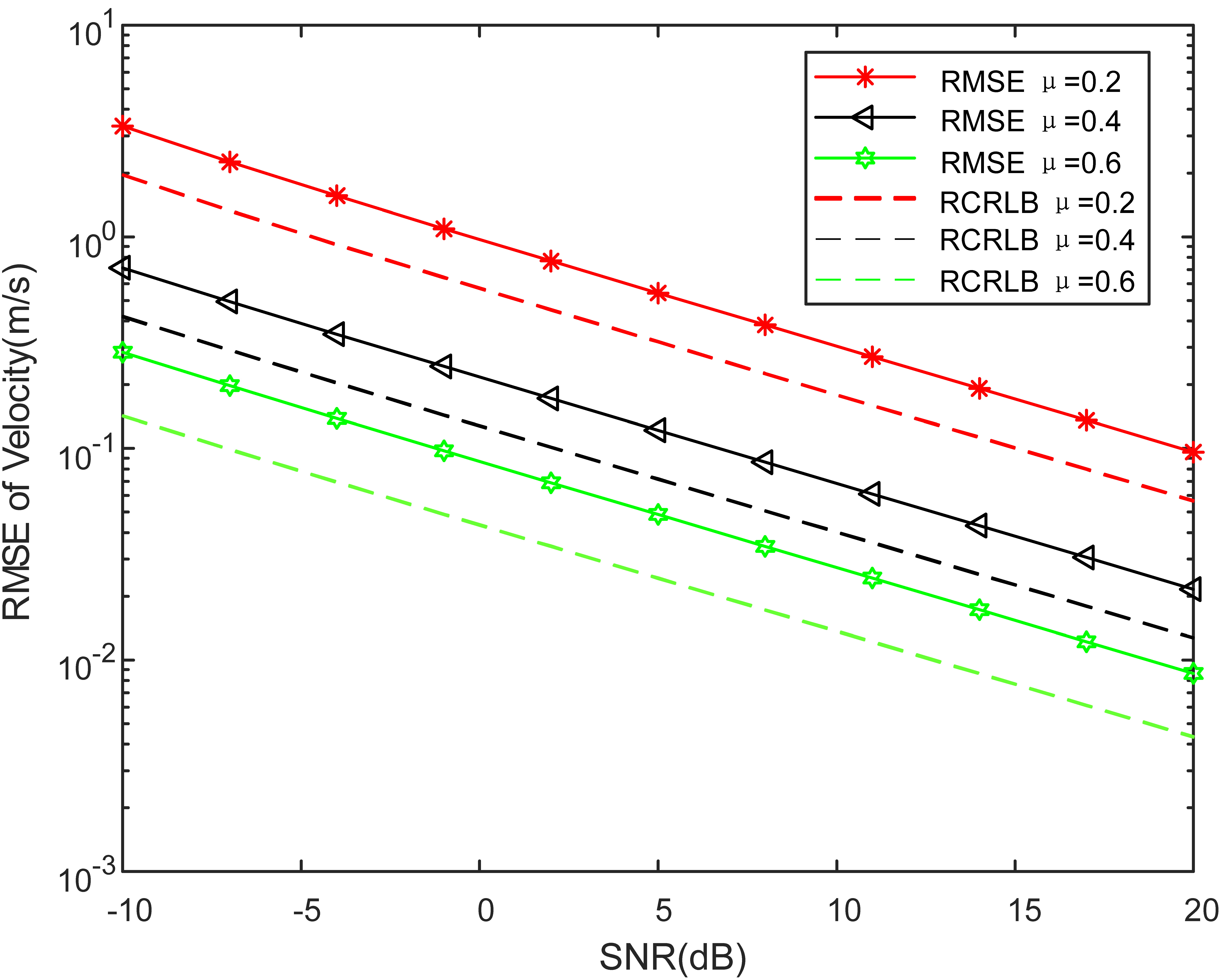}}
\caption{RMSE of velocity for various $\mu$ values together with the corresponding RCRLB values.}
\label{RMSE_nu_vs_SNR}
\end{figure}


\textbf{JRC performance:} Based on the result of communication symbol recovery in (\ref{eq:JRC_B5}), we know that the recovered communication symbols depend on the estimated target position parameters. Thus, communication performance cannot be analyzed separately without considering radar performance here. In this experiment, we take the radar performance and communication performance into account simultaneously, where the communication effectiveness of our proposed OAM-JRC system is assessed through standard metrics, such as BER and communication throughput. The BER metric of our proposed OAM-JRC scheme has been presented in the former section. For the communication throughput metric, it refers to the actual rate of successful data delivery within a given time period \cite{XLiu2023}, defined as
\begin{align}
    S_{\mathrm{succ}}=\frac{N_{\mathrm{succ}}}{T}
\end{align}
where $N_{\mathrm{succ}}$ is the transferred data and $T$ is the transmission time. 

\begin{figure}[t]
\centerline{\includegraphics[scale=0.50]{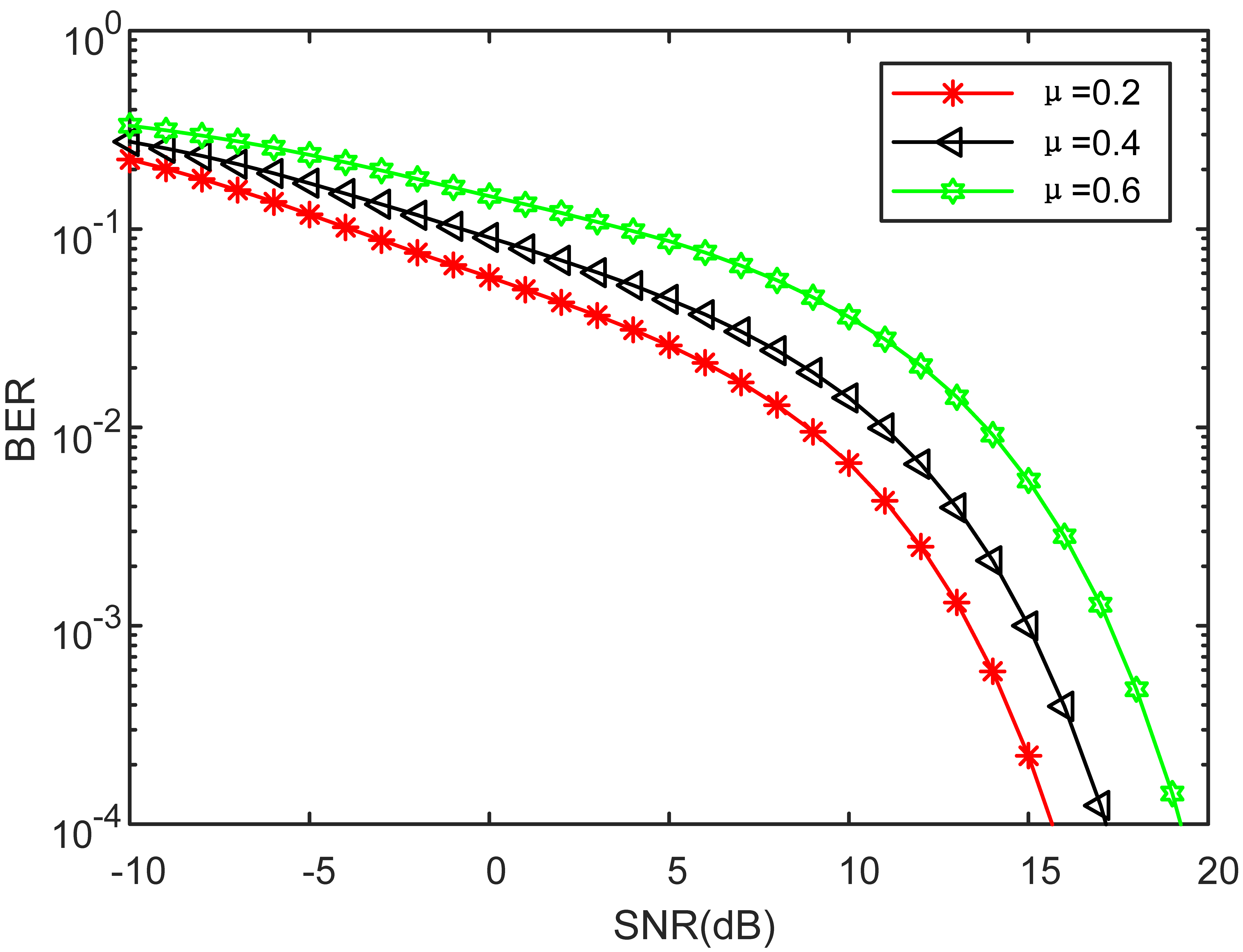}}
\caption{BER of our proposed OAM-JRC system with respect to SNR for different sharing factors.}
\label{BERvsSNR}
\end{figure}

We first analyze the BER performance of the OAM-JRC system with respect to SNR for different sharing factors. As shown in Fig. \ref{BERvsSNR}, the BER of the OAM-JRC system improves with the SNR increasing. When the SNR is fixed, the BER of the OAM system decreases with the decrease of sharing factor value. In order to better display the relationship between radar performance and communication performance, the radar RMSE, BER, and communication throughput with respect to the sharing factor are shown in Figs. \ref{RMSEvsMu} and \ref{BER_ThroughputvsMu}, respectively, where we fix the SNR at 5 dB. As we can see from the figures, the simulation results show a tradeoff between radar performance and communication performance, where one factor is enhanced at the cost of the other parameter. In the region of low sharing factor $\mu$, low BER and high throughput performance can be obtained while the RMSE performance deteriorates. In contrast, in the region of high sharing factor $\mu$, low RMSE performance can be obtained while the BER and throughput performance deteriorate.

\begin{figure}[t]
\centerline{\includegraphics[scale=0.53]{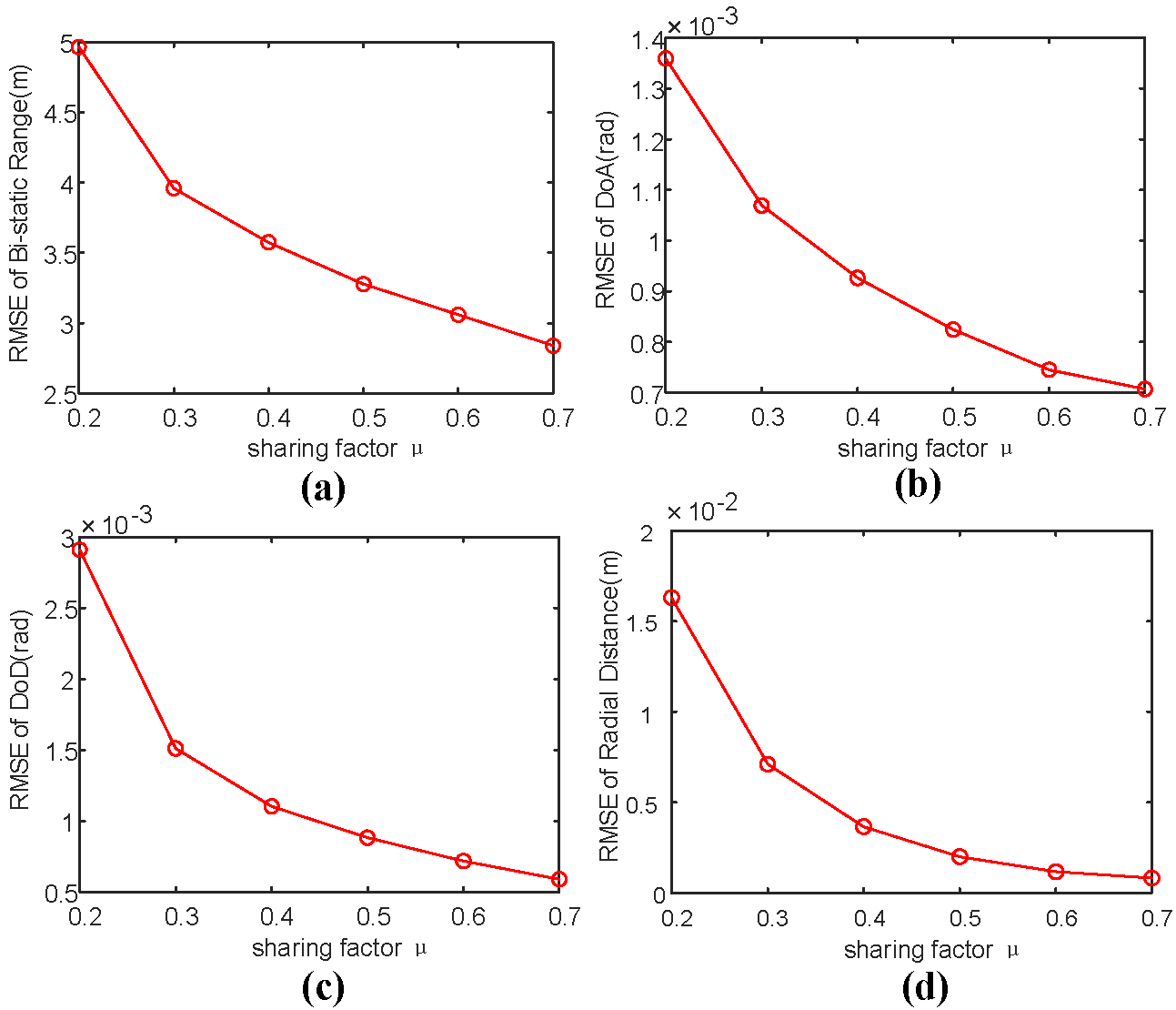}}
\caption{RMSEs of our proposed OAM-JRC system with respect to different sharing factors.}
\label{RMSEvsMu}
\end{figure}
\begin{figure}[t]
\centerline{\includegraphics[scale=0.43]{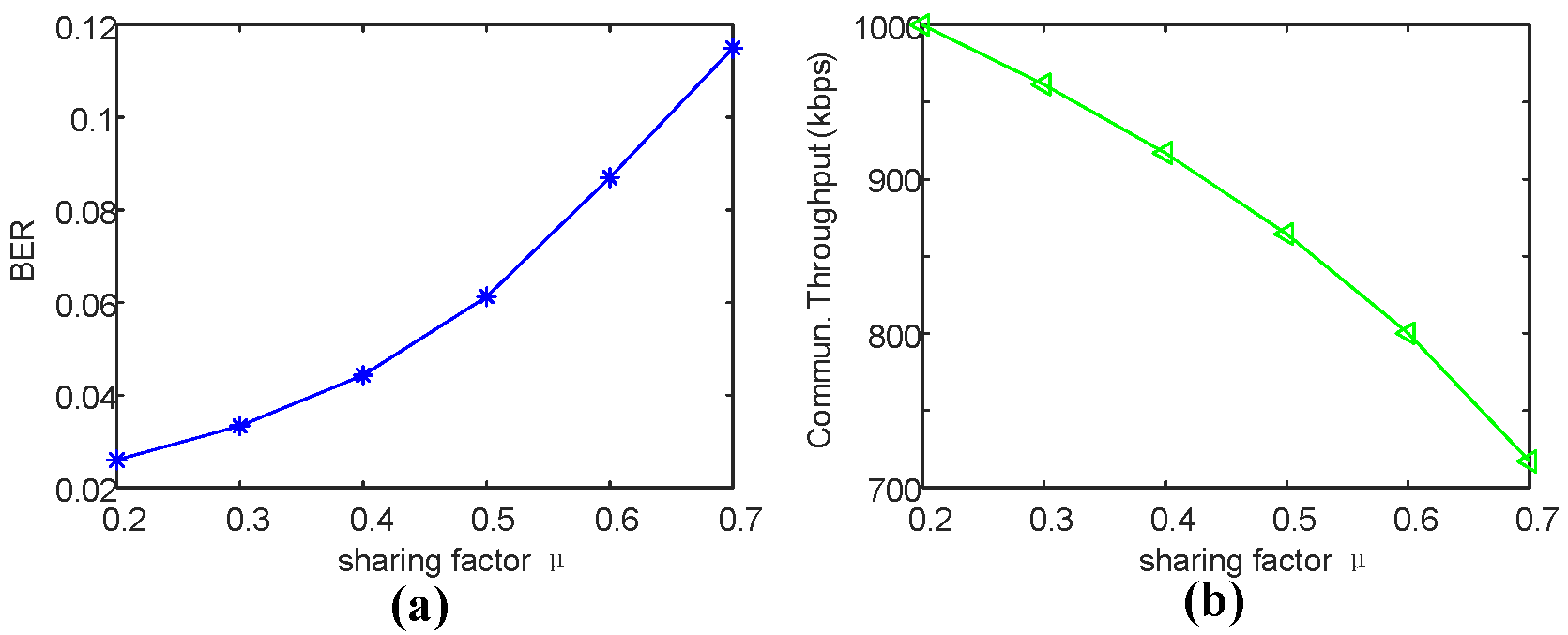}}
\caption{BER and communication throughput of our proposed OAM-JRC system with respect to different sharing factors.}
\label{BER_ThroughputvsMu}
\end{figure}

\section{Summary}       \label{sec:summary}
In this paper, the application of OAM technique in a bi-static system has been investigated which aims to overcome the challenges of autonomous vehicles. In contrast to conventional UCA-based OAM systems employing isotropic antennas, we adopted a spatial modulation scheme utilizing a ULA of traveling-wave antennas, enabling the simultaneous generation of multiple LG vortex beams. In terms of target detection, a radar-only estimation algorithm was developed for target positioning and velocity detection. To ensure the identifiability and reconstruction of JRC parameters, an OAM-based MDM strategy was introduced that operates across radar and JRC frames. In addition, we provided a comprehensive performance analysis of the proposed OAM-based JRC system, including recovery guarantees, CRLB derivations for radar parameters estimation, and BER, throughput evaluations for communication reliability. Both analytical and numerical results validate the effectiveness of the proposed approach and our proposed OAM multiplexing technique allows for a trade-off between radar and communication performance.

\bibliographystyle{IEEEtran}
\bibliography{main}

\end{document}